\documentclass[12pt]{article}
\usepackage{authblk}
\usepackage{geometry}
\usepackage{amsmath}
\usepackage{amssymb}
\usepackage{amsthm}
\usepackage{mathrsfs}
\usepackage{caption}
\usepackage{subcaption}
\usepackage{customcommands}
\usepackage{bm}
\usepackage{Nettastyle}
\usepackage{dsfont}
\usepackage{comment}
\usepackage{enumerate}

\usepackage{amsmath}
\usepackage{amssymb}
\usepackage{amsfonts}

\usepackage{braket}
\usepackage[normalem]{ulem}
\usepackage{color}
\usepackage{bbold}
\usepackage{ulem}
\usepackage{mathtools}
\usepackage{tensor} 
\usepackage{tikz}
\usetikzlibrary{arrows, positioning, quotes}

\newcommand{\bega}{\begin{gather}}
\newcommand{\eega}{\end{gather}}

\newcommand{\bi}{\begin{itemize}}
\newcommand{\ei}{\end{itemize}}
\newcommand{\ben}{\begin{enumerate}}
\newcommand{\een}{\end{enumerate}}
\newcommand{\bca}{\begin{cases}}
\newcommand{\eca}{\end{cases}}
\newcommand{\bln}{\begin{align}}
\newcommand{\eln}{\end{align}}
\newcommand{\bst}{\begin{split}}
\newcommand{\est}{\end{split}}
\def\ie{\begin{equation}\begin{aligned}}
\def\fe{\end{aligned}\end{equation}}
\newcommand{\bma}{\le(\begin{matrix}}
\newcommand{\ema}{\end{matrix}\ri)}

\newcommand{\field}[1]{\ensuremath{\mathbb{#1}}}
\newcommand{\RR}{\field{R}}

\newcommand\sA{{\ensuremath{{\mathcal A}}}}
\newcommand\sB{{\ensuremath{{\mathcal B}}}}
\newcommand\sC{{\ensuremath{{\mathcal C}}}}

\newcommand\sH{{\ensuremath{{\mathcal H}}}}
\newcommand\sK{{\ensuremath{{\mathcal K}}}}

\newcommand\sO{{\ensuremath{{\mathcal O}}}}

\newcommand\sW{{\mathcal W}}
\newcommand\sX{{\mathcal X}}

\newcommand{\bid}{\mathbf{1}}

\preprint{MIT-CTP/5607}
\title{Algebraic ER=EPR and Complexity Transfer}

\author{Netta Engelhardt}
\author{and Hong Liu}

\affiliation{Center for Theoretical Physics, Massachusetts Institute of Technology, \\Cambridge, MA 02139, USA}
\emailAdd{engeln@mit.edu}
\emailAdd{hong\_liu@mit.edu}

\abstract{We propose an algebraic definition of ER=EPR in the $G_N \to 0$ limit, which associates bulk spacetime connectivity/disconnectivity to the operator algebraic structure of a quantum gravity system. The new formulation not only includes information on the amount of entanglement, but also more importantly {\it the structure of entanglement}. We give an independent definition of a quantum wormhole as part of the proposal.
This algebraic version of ER=EPR sheds light on a recent puzzle regarding spacetime disconnectivity in holographic systems with ${\cal O}(1/G_{N})$ entanglement. We discuss the emergence of quantum connectivity in the context of black hole evaporation and further argue that at the Page time, the black hole-radiation system undergoes a transition involving the transfer of an emergent type III$_{1}$ subalgebra of high complexity operators from the black hole to radiation. We argue this is a general phenomenon that occurs whenever there is an exchange of dominance between two competing quantum extremal surfaces.}

\begin{document}

\maketitle

\section{\label{sec:intro}Introduction}

Developments in AdS/CFT over the past two decades have revealed profound connections between the structure of spacetime and entanglement (see~\cite{RyuTak06, HubRan07, Van09, Van10, CzeKar12, MalSus13, Van13, VerVer13a, FauLew13, EngWal14, AlmDon14, DonHar16, CotHay17} among others). A particularly elegant and powerful proposal that highlights such connections is ER=EPR~\cite{MalSus13, Van13, VerVer13a} -- the hypothesis that any 
entangled quantum gravitational systems are connected by some kind of Einstein-Rosen (ER) bridge.

The prototypical example of ER=EPR is an eternal black hole in AdS, which describes a pair of quantum gravity systems in a thermofield double state with a temperature $T > T_{HP}$, where $T_{HP}$ is the Hawking-Page temperature~\cite{HawPag83, Mal01}.  In this example, there is an order ${\cal O}(1/G_N)$ amount of entanglement between left and right systems, and the ER bridge is classical, i.e. the systems are connected in a classical spacetime. See Fig.~\ref{fig:wh1}(a). 
By contrast,  when $T < T_{HP}$, left and right systems have entanglement of order $\sO(G_N^0)$ and are described by two classically disconnected spacetimes with entangled quantum fields. See Fig.~\ref{fig:wh1}(b). 
In this case, it is sometimes said that there is a quantum ER bridge between them; however, this turns out to be misleading: we will show that the most natural notion of quantum wormhole does not include this example. The absence of another independent definition of a quantum wormhole has meant that to date, ER=EPR, as stated above, is more of a slogan than an equivalence.

\begin{figure}[t]
\begin{center}
\includegraphics[scale=0.6]{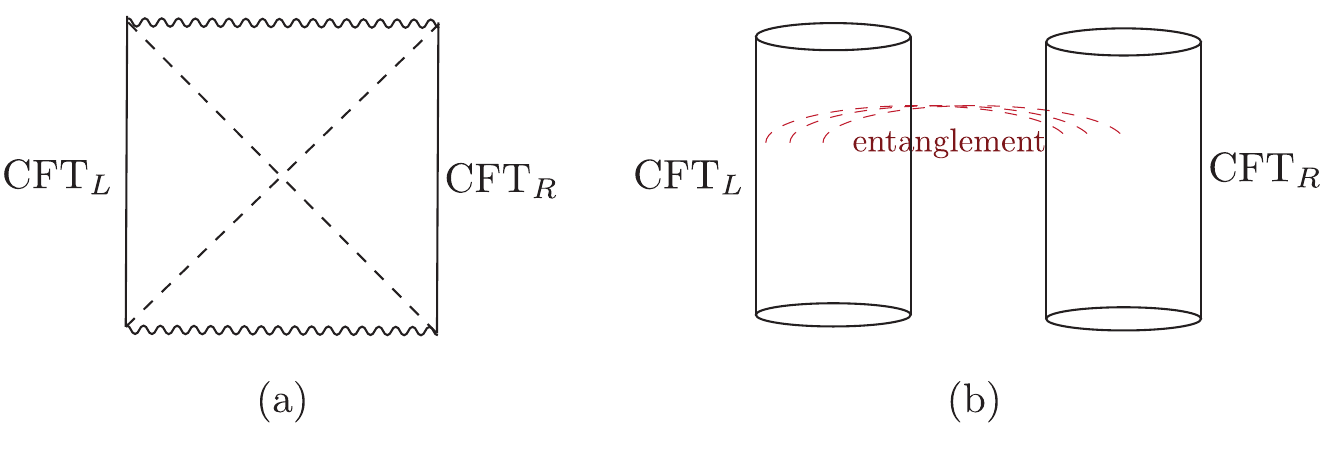} \;\;
\end{center}
\caption{The bulk duals of two copies of a holographic CFT in a thermofield double state. (a) 
For $T > T_{HP}$, the two sides are connected by an Einstein-Rosen bridge. 
(b) For $T < T_{HP}$, we have two classically disconnected spacetimes with entangled quantum fields between them.}
 \label{fig:wh1}
\end{figure}

In this paper we propose a more refined formulation of ER = EPR, called ``algebraic ER=EPR,''  
which associates bulk spacetime connectivity/disconnectivity with the operator algebraic structure of a quantum gravity system.\footnote{See~\cite{NogBan21, BanDor22, BanMor23} for a different algebraic perspective.} 
The operator algebraic structure (in terms of von Neumann algebras) not only includes information on the amount of entanglement, but also more importantly {\it the structure of entanglement}. The new  formulation additionally allows us to propose an independent definition of a quantum wormhole. Since spacetime geometry and geometrical notions such as connectivity are semiclassical concepts, ER=EPR and questions on spacetime structure from entanglement can be sharply formulated only in the $G_N \to 0$ limit ($G_N$ is the Newton constant), which is the  regime that we work with.\footnote{In Euclidean path integral, even for small $G_N$, there can be non-perturbative contributions from summing over spacetimes of different topologies. We will not consider such contributions 
in this paper due to our current lack of systematic Lorentzian understanding of such contributions.}

Our proposal is motivated from three recent developments:

\ben 

\item In~\cite{LeuLiu21a, LeuLiu21b} it was observed that classical connectivity of two systems in a thermofield double state  implies that  the operator algebra of each system is a type III$_1$ von Neumann algebra in the $G_N \to 0$ limit. In other words, classical spacetime connectivity in this context requires a very specific entanglement structure (as characterized by type III$_1$ von Neumann algebras) in the corresponding boundary system.

\item An interesting puzzle regarding ER=EPR was discussed in~\cite{EngFol22} in the context of an evaporating black hole, denoted $B(t)$. 
Consider two times $t_1 < t_P < t_2$ such that   
\be 
 S_{\rm vN} [B (t_1)] = S_{\rm vN} [B (t_2)] \sim \sO(1/G_N)  \ ,
 \ee
where $t_P$ is the Page time~\cite{Pag93b}, and $S_{\rm vN} [B (t)] $ is the von Neumann entropy of the black hole. See Fig.~\ref{fig:count0}. At $t_1$ the black hole system is classically disconnected from the radiation in the bulk interior, 
while at $t_2$ they are connected. See Fig.~\ref{fig:bulkwedges} and Fig.~\ref{fig:count}. 
Ref.~\cite{EngFol22} sharpened this observation by constructing a more straightforward but closely related example, in which both systems are gravitational and the minimal QES is homologous to the boundary. See Appendix~\ref{sec:grpu} for more details on this construction.

This example raises an obvious question: what is the physical difference between $t_1$ and $t_2$ that is responsible for
the differences in their spacetime structure? 

\item There have been some recent discussions~\cite{Wit21b,ChaLon22,ChaPen22,JenSor23,KudLeu23} of understanding the generalized entropies of black holes and of de Sitter
using the crossed product. The crossed product changes the structure of operator algebras of a quantum gravity system and thus the entanglement structure, so it is natural to investigate its implications on spacetime connectivity.

\een

\begin{figure}[t]
     \centering
 \begin{subfigure}[b]{0.49\textwidth}
         \centering
         \includegraphics[width=\textwidth]{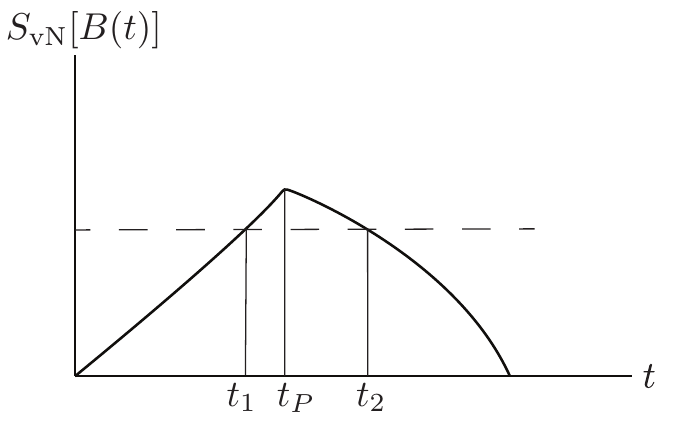}
         \caption{}
                 \label{fig:count0}
     \end{subfigure}
  \begin{subfigure}[b]{0.45\textwidth}
         \centering
         \includegraphics[width=\textwidth]{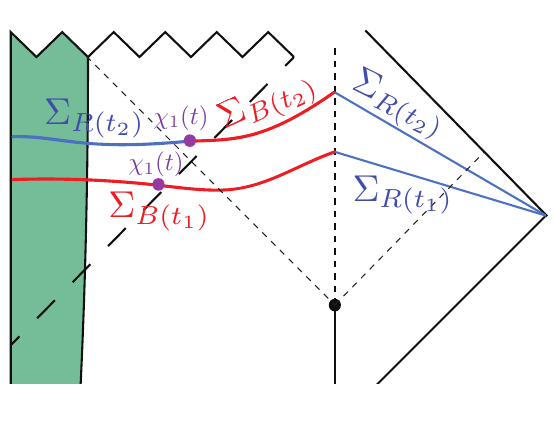}
         \caption{}
                 \label{fig:bulkwedges}
     \end{subfigure}
\caption{(a) The Page curve for an evaporating AdS  black hole describes the time evolution of the von Neumann entropy of the system $B$ dual to the evaporating black hole and the system $R$ consisting of the radiation evaporating into a reservoir. 
The puzzle pointed out in~\cite{EngFol22} 
considers two times $t_1 < t_P < t_2$ that have the same von Neumann entropy of order $\sO (1/G_N)$.
(b) The Penrose diagram of coupled black hole and reservoir systems. 
Cauchy slices for the black hole spacetime at $t_1$ and $t_2$ are shown, with the red regions, labeled by $\Sigma_B$, corresponding to slices of the entanglement wedge of the black hole. At $t_1$, the entanglement wedge of the black hole consists of the full Cauchy slice while that at $t_2$ only the region exterior to $\chi_1$.
Dashed lines represent the positive energy shocks created by coupling the black hole and reservoir.
}
\end{figure}

\begin{figure}
  \centering
     \begin{subfigure}[b]{0.49\textwidth}
         \centering
         \includegraphics[width=\textwidth]{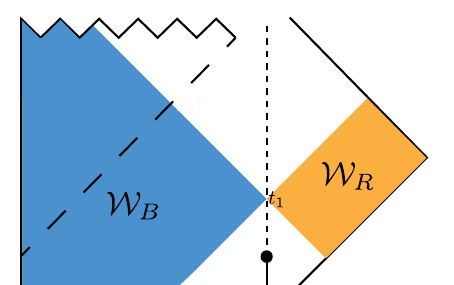}
         \caption{}
                 \label{fig:PrePage}
     \end{subfigure}
     \hfill
     \begin{subfigure}[b]{0.49\textwidth}
         \centering
         \includegraphics[width=\textwidth]{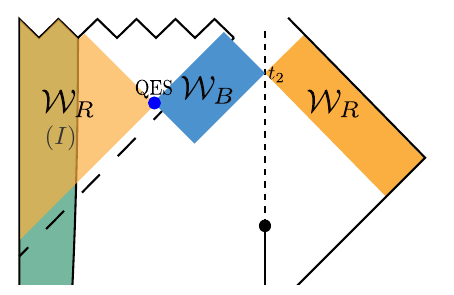}
         \caption{}
              \label{fig:PostPage}
     \end{subfigure}
\caption{(a) At $t_1$, Cauchy slices of the entanglement wedge $\sW_B$ of the black hole system are inextendible, and $\sW_B$ disconnected with the entanglement wedge of the radiation. (b) At $t_2$, the entanglement wedge of the black hole is given by the part labeled as $\sW_B$, while region $I$ is a subset of the entanglement wedge $\sW_R$ of the radiation. They are now classically connected at the quantum extremal surface (QES).  In the Penrose diagram, $\sW_B$ and $\sW_R$ appear to also be connected at the boundary at all times after coupling,  but this is just an artifact of the Penrose diagram: $\sW_B$ and $\sW_R$ are separated by an infinite proper distance there and are not gravitationally interacting.
 }
 \label{fig:count}
\end{figure}

Our proposal incorporates and elucidates upon these developments; it turns out that this crucially requires a  proper definition of classical and quantum spacetimes. As we will see, even in the $G_N \to 0$ limit, there can be 
``quantum'' signatures remaining in a spacetime, which may prevent it from being fully classical. 
For example, in the evaporating black hole example, we argue that at $t_1$, the black hole is a quantum spacetime is quantum connected to the reservoir even though prima facie it appears to be classical and classically disconnected from the reservoir.

The contrast between quantum and classical connectivity between the black hole and the radiation at times $t_1$ and $t_2$ raises another question. It is well-established that the Page time is (1) the time at which the von Neumann entropy of the radiation reaches a global maximum, and subsequently~\cite{Pag93b, HayPre07, Pag13} (2) the time at which the black hole becomes ``transparent'', so that information about the infalling matter can be recovered in the radiation; we now make the observation that (3) it is time at which, in the holographic picture, the radiation and the black hole entanglement wedges become classically connected in the bulk interior. It is therefore natural to ask:

\medskip

{\it What really happens microscopically at the Page time that leads to the confluence of these dynamical changes?
} 

\medskip

\begin{figure}[t]
\begin{center}
\includegraphics[scale=1]{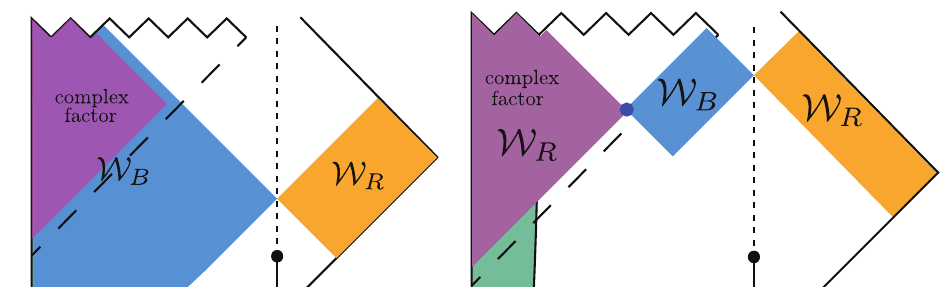} 
\end{center}
\caption{A cartoon of the dynamical transition at the Page time. A type III$_1$ factor consisting of operators of high complexity (Python's lunch~\cite{BroGha19, EngPen21a, EngPen21b}) is transferred from the black hole system to the radiation system at the Page time.}
 \label{fig:page}
\end{figure}

The new and algebraic understanding of connectivity of quantum gravitational systems also brings new insights into the above question. 
We will show in Sec.~\ref{sec:trans} that at the Page time, the system undergoes a dynamical transition in terms of transfer of a particular type III$_1$ subalgebra consisting of operators of high complexity from the black hole to its radiation.  See  Fig.~\ref{fig:page}. This is motivated by the switch between dominant quantum extremal surfaces~(QESs) at the Page time. While QESs are frequently regarded as devices for computing the von Neumann entropy via the ``QES formula''~\cite{EngWal14}, they also define the entanglement wedge~\cite{Wal12, CzeKar12, HeaHub14} and can be used to compute reconstruction complexity~\cite{BroGha19, EngPen21a, EngPen21b}. We propose that switchovers of dominance between QESs can be generically understood at a more fundamental level as the transfer of a subalgebra of high complexity operators between a system and its complement in the large-$N$ limit, in keeping with the quantum information theoretic arguments of~\cite{HarHay13}.
Such a transfer can be used to explain the differences between the structure  of the black hole system at $t_1$ and $t_2$, and the transfer of quantum information from black hole to radiation.

Our main results may be organized into the following list:

\ben 

\item We formulate the following algebraic ER=EPR proposal for a bipartite state on two complete asymptotic boundaries $R_1$ and $R_2$: 

{\it 
Consider two entangled systems $R_1$ and $R_2$ in a pure state $\ket{\psi_{R_1 R_2}}$ with a gravitational bulk dual, which we consider
in the $G_N \to 0$ limit.\footnote{We assume throughout that the bulk dual to $\ket{\psi_{R_1 R_2}}$ has a sensible $G_N \to 0$ limit.}
Denote the operator algebras of $R_1$ and $R_2$ in the large $N$ limit respectively as $\sA_{R_1}$ and $\sA_{R_2}$, which are von Neumann algebras.  Then the entanglement wedge of $R_1 \cup R_2$ is 
\bea
\label{a1}
&\text{disconnected:}  \quad \text{iff $\sA_{R_1}$ and $\sA_{R_2}$ are type I}  \ ; \\
\label{a2}
&\text{connected by a classical wormhole (classical bulk):}   \quad \text{iff $\sA_{R_1}, \sA_{R_2}$ are type III$_1$} 
; \\
& \text{connected by a quantum wormhole ( quantum bulk):}  \quad \text{ 
$\sA_{R_1}$, $\sA_{R_2}$ are not type I}.  
\label{a3}
\eea
}
\noindent The notions of classical and quantum connectivity \textit{require} a specification of the underlying spacetime as classical or quantum\footnote{Classical connectivity only makes sense in a classical spacetime.}, and as we explain explicitly in the main text, this property is something that can in principle be diagnosed intrinsically from the boundary. Note that in~\eqref{a3} we have ``if'' rather than ``iff'' as these are our \textit{definitions} of ``quantum connected''. 
Equation~\eqref{a2} is a generalization of~\cite{LeuLiu21a, LeuLiu21b}. ~\eqref{a1}-\eqref{a3} can be motivated in the following sense. 
When $\sA_{R_1}, \sA_{R_2}$ are type I, then there exists a factorization of the Hilbert space $\sH = \sH_{R_1} \otimes \sH_{R_2}$ with $\sH_{R_1}, \sH_{R_2}$ associated with $R_1$ and $R_2$. But for type II and III algebras, such a factorization does not exist, and thus $R_1$ and $R_2$ can be intuitively considered to be connected.

\item We generalize the above bipartite algebraic ER=EPR to multipartite case, using the bulk dual of canonical purification~\cite{EngWal18}.

\item We use algebraic ER=EPR and its multipartite version to explain how different subsystems of an evaporating black hole system are connected to one another. In addition to the connectivity between the black hole and radiation considered here, we also discuss how the island is connected to the radiation. 

\item We argue that the crossed products discussed in~\cite{Wit21b, ChaLon22, ChaPen22, JenSor23, KudLeu23} lead to quantum  connectivity. 

\item We make a general proposal that switchovers in the dominant QES generically originate in the dynamical transfer of a type III$_1$ factor consisting of operators of high complexity. We further argue for such a transfer based on the principle that the transition should not involve any non-analytic behavior in low energy physical observables. 

For the evaporating black hole, we show that this mechanism of transfer of complexity can be used to explain the differences of connectivity between the black hole and radiation at times $t_1$ and $t_2$, and the transfer of high complexity quantum information from black hole to radiation first {discussed} in~\cite{HarHay13}.

\een

\bigskip

\noindent {\bf Assumptions, Conventions, and Notation}

\bigskip
\noindent We work in the large-$N$ limit, large-$\lambda$ limit of AdS$_{d+1}$/CFT$_{d}$, and we will specifically work to leading order in 
the $G_{N} \to 0$ limit, which is equivalent to the leading order of the $N \to \infty$ limit.\footnote{Those concerned with a discussion of evaporating black holes at leading order in the $G_{N}$ expansion should defer their concerns to Sec.~\ref{sec:quan1}, where we discuss this in detail. We generally work on a moment of time slice in the adiabatic regime; if time evolution is necessary, we will work at finite $N$ and then take the $N\rightarrow \infty$ limit, keeping the leading order behavior only.} 
Here we list our notation, definitions, and assumptions about the bulk and boundary theories. 

\noindent On the bulk side:
\begin{itemize}
    \item We will use ${\cal W}_{R}$ to denote the entanglement wedge of a boundary subsystem $R$. This will typically denote the standard entanglement wedge including the homology constraint. We will occasionally relax the homology constraint to discuss islands; when we do so we will say so explicitly.      
    \item We assume the QES formula~\cite{EngWal14} and the strong Python's Lunch conjecture~\cite{BroGha19, EngPen21a, EngPen21b}, to be reviewed in Sec.~\ref{sec:complex}. 
  
    \item We assume subalgebra-subregion duality~\cite{LeuLiu22} (see also~\cite{Har16, KanKol18, Fau20, KanKol21, GesKan21, FauLi22}): the boundary subalgebra ${\cal A}_{R_{i}}$ obtained in the large-$N$ limit is identified with 
    the bulk operator algebra in the entanglement wedge ${\cal W}_{R_i}$ of $R_i$. It is important to note that ${\cal A}_{R_{i}}$ depends on the (semiclassical) state of the system, which matches with state-dependence of the entanglement wedge depends on the on the gravity side.

      \item By $\mathscr{I}$ we mean a single complete connected component of the asymptotic AdS boundary. 
\end{itemize}

\noindent  On the boundary:

\begin{itemize}

\item By operator algebra we mean a von Neumann algebra. 

    \item By a `state', we mean either a vector or a density matrix.
    \item We say that a state is bipartite on two subsystems $R_1$, $R_2$ if it is a pure state on those two  subsystems: $\ket{\psi_{R_{1}R_{2}}}$. 
    \item We say that a boundary region $R$ has an edge if its boundary domain of dependence $D[R]$\footnote{$D[R]=D^{+}[R]\cup D{-}[R]$, where $D^{+}[R]$ ($D^{-}[R]$) is the set of points $p$ such that every past-directed (future-directed) curve from $p$ crosses $R$.} is continuously extendible in the boundary spacetime (i.e. it is a proper subset of a single connected asymptotic boundary). Otherwise it is edgeless.\footnote{As defined in~\cite{Wald}, the edge of an achronal closed set $S$ is the set of points $p$ in $S$ such that every open neighborhood $O$ of $p$ contains a point in the future of $p$, a point in the past of $p$, and a timelike curve connecting those two points which does not intersect $S$. An edgeless set is a set without an edge.}

    \item We will say that a state is multipartite if it is not bipartite.

    \item Given a CFT state $\ket{\psi_{R_{1}\cdots R_{n}}}$ for boundary subsystems $R_i, i=1,2, \cdots, n$, we use 
     ${\cal A}_{R_{i}}$ to denote the operator algebra associated with $D[R_{i}]$ in the large $N$ limit. 
   
\end{itemize}

\bigskip

\noindent {\bf Plan of the paper}

\bigskip

The plan of the paper is as follows.  In Sec.~\ref{sec:er}, we give a general formulation of the algebraic ER=EPR, including a discussion of various quantum connected examples. 
In Sec.~\ref{sec:multi} we generalize the proposal to multiple-partite case, as well as a new understanding of islands. In Sec.~\ref{sec:trans}, we elaborate on the dynamical transition at the Page time. 
We conclude in Sec.~\ref{sec:diss} with a discussion of apparent counterexamples and some future perspectives.

\section{Algebraic ER=EPR} \label{sec:er}

In this section we formulate our algebraic ER=EPR proposal for bipartite edgeless states.

\subsection{Setup} \label{sec:cvst}

Consider two copies of the boundary theory in a pure state $\ket{\psi_{R_{1}R_{2}}}$, with $R_{1,2}$ denoting the corresponding boundary systems, which will be complete asymptotic boundaries. Suppose $\ket{\psi_{R_{1}R_{2}}}$ has a semiclassical description, 
described by a bulk spacetime $\sW_{R_1 \cup R_2}$ with boundary $R_1 \cup R_2$. 
$\sA_{R_1}$ and $\sA_{R_2}$ are respectively the relevant operator algebras of $R_1$ and $R_2$ in the large $N$ limit. 
We assume that both $\sA_{R_1}$ and $\sA_{R_2}$ are factors and 
\be
\sA_{R_1}' = \sA_{R_2} ,
\ee
where $'$ denotes commutant. We use $\sH_{\rm bulk}$ to denote the bulk Hilbert space around $\sW_{R_1 \cup R_2}$ which is identified with the boundary Hilbert space around $\ket{\psi_{R_{1}R_{2}}}$ in the large $N$ limit. We will use them interchangeably. Our goal is to formulate  an intrinsic boundary description of connectivity of $\sW_{R_1 R_2}$.

Before delving into technical definitions, it is worth reiterating that the notion of ``spacetime connectivity'' is an emergent concept valid in the semiclassical regime, i.e.  in the $G_{N}\rightarrow 0$ limit.  At general finite $G_N$ (i.e. finite $N$), gravitational systems  have large 
fluctuations in the geometry, and there does not appear to exist a precise definition of ``connectivity''. While naively the $G_N \to 0$ limit should give a classical spacetime, one of the lessons of the past few years starting with~\cite{Pen19, AEMM} is that even in this limit, certain quantum signatures can remain in the spacetime, which may prevent the spacetime from being fully classical.

We will therefore begin our formulation of ER=EPR with a definition of classical and ``quantum'' spacetimes:

\begin{enumerate}

\item We say that a state describes a {\it classical spacetime} if there exists a family of Cauchy slices with boundary time ranges of order $O(G_N^0)$,
and a coordinate basis on each member of this family, 
in which the metric components are all of order ${\cal O}(G_N^0)$, and furthermore in the $G_N \to 0$ limit, 
 
 \ben 
 
\item  fluctuations of diffeomorphism invariants go to zero as ${\cal O}(G_{N}^{a})$ for some $a>0$;\footnote{This is a bit more general than the usual semiclassical description where the fluctuations are of order ${\cal O}(\sqrt{G_N})$.} 

\item 
spacetime diffeomorphism invariants such as volumes, areas, and lengths  on this family 
do not scale as $O(G_{N}^{-a})$ with $a > 0$.
\een

\item  We say that a state describes a {\it quantum volatile} spacetime if there is a coordinate  basis where the metric components are of order ${\cal O}(G_N^0)$, and either or both of the two conditions in item 1 are violated, but fluctuations of geometric invariants are suppressed in some small parameter $\epsilon \sim {\cal O} (G_N^0)$. In the case that only (b) is violated, such an $\epsilon$ does not have to exist.

\item If  spacetime fluctuations of a state are unsuppressed, i.e. of order ${\cal O} (G_N^0)$ but with no other suppression parameter, or scaling with $G_N$ with a negative power, we say it has {\it no geometric} description.

\end{enumerate}

We will assume that quantum volatility of the bulk can be diagnosed from properties of boundary quantities. That is, there exist CFT quantities that can be used to probe the magnitudes or fluctuations of bulk geometric invariants. 
Under this holographic assumption, we define the \textit{``the classical condition''} as follows: if the bulk  invariants  
computed from the boundary quantities do not scale with $G_{N}^{a<0}$ or have vanishing fluctuations, the dual CFT state satisfies the classical condition.

Examples of such boundary quantities include various two-sided observables such as two-point functions (i.e. with two operators in $R_1, R_2$ respectively), which can be used to probe geodesic length, or the complexity=volume conjecture, which can be used to probe bulk volume.

\subsection{Algebraic ER=EPR proposal} \label{sec:prop}

We now present our proposal of algebraic ER=EPR, starting with the boundary description of a disconnected bulk spacetime. First recall the example of two copies of CFT in the thermofield double state at $T < T_{HP}$. 
In this case, the bulk spacetime consists of two disconnected global AdS spacetimes. We can perform canonical quantization 
in each of them, which results in the bulk Hilbert space $\sH_{\rm bulk} = \sH_R^{\rm (Fock)} \otimes \sH_L^{\rm (Fock)}$.  The bulk state is described by an entangled state between $\sH_R^{\rm (Fock)}$ and $\sH_L^{\rm (Fock)}$. In the large $N$ limit, the boundary algebras $\sA_R, \sA_L$ of CFT$_R$ and CFT$_L$ are type I, dual respectively to $\sB (\sH_R^{\rm (Fock)})$ and $\sB (\sH_L^{\rm (Fock)})$.\footnote{$\sB (\sH)$ denotes the algebra of bounded operators on $\sH$.}  

This example motivates us to propose the following general statement  for $R_1, R_2$ in a pure state describing a classical $\sW_{R_1 \cup R_2}$,
\be \label{disp}
\text{\it $\sA_{R_1}$ and $\sA_{R_2}$ are type I if and only if $\sW_{R_1 \cup R_2}$ is disconnected.} 
\ee
When a classical spacetime $\sW_{R_1 \cup R_2}$ is disconnected, it means that $\sW_{R_1}$ and $\sW_{R_2}$ are each a 
complete spacetime. In this case, in the $G_N \to 0$ limit, given that no geometric invariants scale with $G_{N}^{a<0}$, 
we can build a Fock space on $\sW_{R_1}$ (similarly $\sW_{R_2}$) from the standard procedure of canonical quantization of matter fields and metric perturbations. We then have $\sH_{\rm bulk} = \sH_{R_1}^{\rm (Fock)} \otimes \sH_{R_2}^{\rm (Fock)}$, and 
$\sA_{R_1} = \sA_{\sW_{R_1}} = \sB (\sH_{R_1}^{\rm (Fock)})$ is type I (as is $\sA_{R_2}$). 
Conversely, when $\sA_{R_1}$ and $\sA_{R_2}$ are type I, the bulk Hilbert space $\sH_{\rm bulk}$ can be factorized into 
a tensor product of those associated $R_1$ and $R_2$, i.e. $\sH_{\rm bulk} = \sH_{R_1} \otimes \sH_{R_2}$. When the bulk spacetime is classical, $R_1$ and $R_2$ are then boundaries of disconnected components as otherwise the Hilbert space cannot be factorized. 
When $\sW_{R_1 \cup R_2}$ is quantum, we can use $\sA_{R_1}$ and $\sA_{R_2}$ being type I as a {\it definition} of a disconnected bulk.

Statement~\eqref{disp} immediately implies that
\be 
\text{\it $\sA_{R_1}$ and $\sA_{R_2}$ are type II or III if and only if $\sW_{R_1 \cup R_2}$ is connected.} 
\ee
Depending on whether the dual spacetime $\sW_{R_1 \cup R_2}$ is classical or quantum volatile we can 
further separate whether the bulk spacetime is classically connected or quantum connected. To do so, we must first give a precise definition of these notions.  

When $\sW_{R_1 \cup R_2}$ is classical, it can be connected in the usual sense that there exists a continuous path whose length does not diverge with $G_{N}^{a<0}$ from $R_1$ to $R_2$, or disconnected if such a path does not exist. 
Classical connectivity can thus be defined for classical spacetimes as follows: 

\paragraph{Definition:} \textit{Let $\sW_{R_{1}\cup R_{2}}$ be a classical spacetime. It has a classical wormhole connecting $R_{1}$ to $R_{2}$ (equivalently, is classically connected) if there exists a continuous spacelike path $\gamma$ from $\sW_{R_{1}}$ to $\sW_{R_{2}}$ that (1) lives in $\sW_{R_{1}\cup R_{2}}$ and (2) has a length that does not scale with $G_{N}^{a<0}$. 
}

\vspace{0.3cm}

In Sec.~\ref{sec:class} we will argue that, for a classical spacetime, classical connectivity of $\sW_{R_{1}\cup R_{2}}$ is equivalent to ${\cal A}_{R_{1}}$ and ${\cal A}_{R_{2}}$ being type III$_{1}$. 
The basic idea comes from the observation in~\cite{LeuLiu21a,LeuLiu21b} that two systems in a thermofield double state are classically connected only if 
the operator algebra of each system is a type III$_1$ von Neumann algebra in the $G_N \to 0$ limit.
This observation can readily be applied to general bipartite situations with two systems $R_1, R_2$ together in an entangled pure state. 
Intuitively, if the dual spacetime $\sW_{R_{1} \cup R_{2}}$ has a classical wormhole connecting $R_{1}$ to $R_{2}$, then the entanglement wedges $\sW_{R_1}, \sW_{R_2}$ for $R_1$ and $R_2$ should both have a boundary in the interior of the spacetime, and we furthermore expect that there exists a Cauchy slice of the dual spacetime on which this boundary is just the shared QES for $R_1, R_2$. 
In the $G_N \to 0$ limit, the bulk theory reduces to a quantum field theory on $\sW_{R_{1}\cup R_{2}}$, and the algebra associated with a 
 proper subregion is type III$_{1}$~\cite{Ara64, Fre84, Haag, Wit18}. 
Invoking subalgebra-subregion duality of~\cite{LeuLiu22}, we find that the boundary algebra ${\cal A}_{R_{i}}$ for $R_i$ is likewise type III$_{1}$. We thus expect that when the spacetime is connected as a classical geometry, the corresponding boundary algebras are type III$_{1}$.

For quantum volatile spacetimes, we need a notion of quantum connectivity.\footnote{Later we will also consider cases involving connectivity between a classical entanglement wedge and a quantum volatile entanglement wedge.}
 Intuitively, this should be some failure of factorization between ${\cal A}_{R_{1}}$ and ${\cal A}_{R_{2}}$ in the $G_{N}\rightarrow 0$ limit. A clear way of formalizing this expectation is to demand that $\sH_{\rm bulk}$  cannot be factorized. 
This is simply the requirement that the algebra not be type I. We therefore define quantum connectivity as follows:

\paragraph{Definition:} \textit{Let $\sW_{R_{1}\cup R_{2}}$ be a quantum volatile spacetime. It has a quantum wormhole connecting $R_{1}$ to $R_{2}$ (equivalently, is quantum connected) if ${\cal A}_{R_{1}}$ and ${\cal A}_{R_{2}}$ are not type~I. 
}

In Sec.~\ref{sec:quan1} we will argue that in the example an evaporating black hole before the Page time (i.e. at $t_1$ of Fig.~\ref{fig:count0}), where the black hole is seemingly classically disconnected from radiation, they are in fact quantum connected. 
Other examples of quantum connectivity are discussed in Sec.~\ref{sec:quan2} and Sec.~\ref{sec:quan3}, which include those whose bulk operator algebras are obtained by the crossed product and the thermal field double state at time $O(1/G_N)$.

Collecting all of the above elements, we have the following full algebraic ER=EPR proposal:

\paragraph{Proposal:} \textit{Let $\ket{\psi_{R_{1}R_{2}}}$ be a pure bipartite state on two copies of the static cylinder (so that $R_{1}, R_{2}$ are edgeless), and let $\sW_{R_1 \cup R_2}$ be its semiclassical dual description in the large-$N$ limit. Then: 
\begin{enumerate}[I.]
 \item $\sW_{R_1 \cup R_2}$ is disconnected if  and only if ${\cal A}_{R_{1}}$ and ${\cal A}_{R_{2}}$ are both type I.  
	\item $\sW_{R_1  \cup R_2}$ has a classical wormhole connecting $R_{1}$ to $R_{2}$ if and only if ${\cal A}_{R_{1}}$ and ${\cal A}_{R_{2}}$ are both type III$_{1}$ and $\ket{\psi_{R_{1}R_{2}}}$ satisfies the classical condition. 
 	\item  $\sW_{R_1  \cup R_2}$ has a quantum wormhole connecting $R_{1}$ to $R_{2}$ 
	if and only if  $\ket{\psi_{R_{1}R_{2}}}$ fails the classical condition and $\sA_{R_1}, \sA_{R_2}$ are not type I. 
\end{enumerate}}

We will now further elaborate on item II. {We will also discuss a provision for a special case where $\sW_{R_{1}\cup R_{2}}$ is quantum volatile, but by acting on just one of the $R_{i}$ with a unitary, it is possible to make the spacetime classical. In this case we will argue that it is still reasonable to discuss classical connectivity, and we will give an appropriate diagnostic thereof.} We close with examples for item III.

\subsection{Connectivity of a classical spacetime} \label{sec:class}

We now give a more detailed justification of item II (classical connectivity) of algebraic ER=EPR.

\vspace{0.3cm}

In arguing that the type III$_{1}$ algebra of one subsystem of a bipartite edgeless state is equivalent to the existence of a classical wormhole as defined above when the spacetime is classical,  
we will use the following lemma:

\begin{lem} There exists a Cauchy slice of $\sW_{R_{1} \cup R_{2}}$ which is a union of a Cauchy slice of $\sW_{R_{1}}$ and a Cauchy slice of $\sW_{R_{2}}$. 
\end{lem}

\begin{proof}
Because $\ket{\psi_{R_{1} R_{2}}}$ is a pure state, $\sW_{R_{1} \cup R_{2}}$ has no boundary in the interior of the spacetime and the QES defining $\sW_{R_{1}}$ is identical to the QES defining $\sW_{R_{2}}$\footnote{If there are competing QESs of identical entropy to leading order, we must pick the same one for the two subsystems.}. Let $\Sigma_{i}$ be Cauchy slices of the $\sW_{R_{i}}$ where $i=1,2$; because the QES is the same for $R_{1}$ and $R_{2}$, $\partial \Sigma_{1}=\partial \Sigma_{2}$.\footnote{This also includes the possibility that the $\sW_{R_{i}}$ are disconnected 
in which case $\partial \Sigma_{1}=\partial \Sigma_{2} = \varnothing$.} Define $\Sigma=\Sigma_{1}\cup \Sigma_{2}$. Then $\Sigma$ has no boundary, so it is a Cauchy slice of $\sW_{R_{1}\cup R_{2}}$.\end{proof}

We will now give the argument identifying type III$_{1}$ with $G_{N}$-finite pathwise connectedness. By subalgebra-subregion duality, ${ \cal A}_{\sW_{R_{i}}} = \sA_{R_i}$ where ${ \cal A}_{\sW_{R_{i}}}$ denotes 
the algebra of the bulk quantum fields on $\sW_{R_{i}}, i=1,2$.

Suppose  $\sW_{R_{1} \cup R_{2}}$ connects $R_{1}$ to $R_{2}$, i.e. there exists a continuous path (whose length does not scale with powers of $G_N$) from $R_{1}$ to $R_{2}$ through $\sW_{R_{1} R_{2}}$. Since $\sW_{R_{1}}\cap\sW_{R_{2}}=\varnothing$, the path must exit $\sW_{R_{1}}$ and enter $\sW_{R_{2}}$. So $\partial \sW_{R_{i}}\neq \varnothing$. By global hyperbolicity of the domain of dependence, $\partial \Sigma_{i}\neq \varnothing$. So there is an interior boundary to each entanglement wedge, i.e. $\sW_{R_i}$ are proper subregions. Their corresponding subalgebras should then be type III$_{1}$ algebras.

Now consider the converse. Suppose that the $\sA_{R_i}$ are 
type III$_{1}$, and thus do not admit a definition of a trace. 
From our earlier discussion of disconnected case, this in turn means that 
$\sW_{R_i}$ cannot contain a complete Cauchy slice of the spacetime: it must have a boundary in the spacetime interior. 
By global hyperbolicity of a domain of dependence, every Cauchy slice $\Sigma_{i}$ of $\sW_{R_i}$ has  the same topology, and thus every $\Sigma_{i}$ has a boundary in the interior of the spacetime. Since $\ket{\psi_{R_{1}R_{2}}}$ is a pure state,  Cauchy slices of $\sW_{R_{1} R_{2}}$ have no boundary. Since Cauchy slices of the $\sW_{R_i}$ \textit{do} have a boundary, and by Lemma 1 there exists a Cauchy slice $\sW_{R_{1} \cup R_{2}}$ which is a union of the Cauchy slices of the $\sW_{R_i}$, $\partial \Sigma_{1}=\partial \Sigma_{2}\in \sW_{R_{1}  \cup R_{2}}$. Since the boundaries $\partial \Sigma_{i}$ are nonempty, there exists a path on $\Sigma$ from $\sW_{R_{1}}$ to $\sW_{R_{2}}$, and thus by assumption there also exists a $G_{N}$-independent length path i.e. $\sW_{R_{1}  \cup R_{2}}$ is connected.\footnote{Note that in this argument we only used that $\sA_{R_i}$ does not have a trace, which applies to any type III algebras. The requirement of type III$_1$ comes from self-consistency: the bulk algebra $\sW_{R_i}$
in a region with an interior boundary must be type III$_1$ from the standard results of QFT in a curved spacetime.}

\subsection{In-Between Classical and Quantum Volatile Spacetimes} \label{sec:inbetween}

Suppose now that our spacetime $\sW_{R_{1}\cup R_{2}}$ is quantum volatile, but this is purely a consequence of the behavior of one entanglement wedge? That is, suppose that $\sW_{R_1}$ is classical but $\sW_{R_{2}}$ is quantum volatile? Given that $\sW_{R_{1}}$ will have a standard, non-fluctuating QES (which of course may be empty), it seems intuitive that it may still be possible to define classical connectivity if $\sW_{R_2}$ is quantum volatile due to $G_{N}$--dependent divergences in geometric invariants on the relevant family of Cauchy slices (rather than fluctuations in the geometry). In this case  a notion of classical connectivity can still be defined by asking if there exists a $G_{N}$-independent length path from $\sW_{R_{1}}$ to some point in $\sW_{R_2}$ that never exits $\sW_{R_{1}\cup R_{2}}$. Is there a boundary way to diagnose this special way of being quantum volatile? 

The aim of such a diagnostic is to determine if the QES of $\sW_{R_{2}}$ is non-fluctuating. A simple protocol for this is to apply a unitary $\mathbb{I}_{R_{1}}\otimes U_{R_{2}}$ which does not modify $\sW_{R_{1}}$. If the dual to the new state $\mathbb{I}_{R_{1}}\otimes U_{R_{2}}\ket{\psi_{R_{1}R_{2}}}$ is classical, then we may use item II to diagnose its connectivity. 

Does such a unitary always exist? Fortunately, the answer is yes: if the so-called \textit{canonical purification} of $\psi_{R_{1}}$ is classically connected, then there exists a path from $R_{1}$ to some point in $\sW_{R_{2}}$ whose length does not diverge with $G_{N}$. In this case and this case only, we say that $R_{1}$ and $R_{2}$ are classically connected in the state $\ket{\psi_{R_{1}R_{2}}}$ even though $\sW_{R_{2}}$ is quantum volatile. In Sec.~\ref{sec:multi} we will review the canonical purification and its gravitational dual, which is a simple CPT-conjugation of $\sW_{R_{1}}$ about its dominant QES.

\subsection{Examples of quantum connectivity} \label{sec:quan1}

We consider various examples of quantum connectivity. Since ours is the first rigorous independent definition of quantum connectivity, these examples serve as motivation for item III in lieu of a justification via equivalence to an extant definition.  

\subsubsection{Evaporating black holes}

We now examine in detail the evaporating black hole example discussed in the introduction, and argue that black hole and radiation are quantum connected.

It is worth stressing that we consider the bulk theory in the $G_N \to 0$ limit and focus on leading order in the small $G_N$ expansion.  
While the evaporation of a black hole involves a time range of order $O(G_N^{-1})$, and the scrambling time is of order ${\cal O} (\log G_N^{-1})$, we will only be interested in describing the states of the system at various but fixed times, rather than a dynamical description of the full history. For such a purpose, focusing on the leading term in a perturbative $G_N$ expansion is adequate.  Since the evaporation of the black hole is slow, at any given time only physics at time scales of order ${\cal O}(G_N^0)$ will be relevant, during which the system can be treated as being quasi-static and in quasi-equilibrium.

\begin{figure}
    \centering
\begin{subfigure}[b]{0.49\textwidth}
         \centering
         \includegraphics[width=\textwidth]{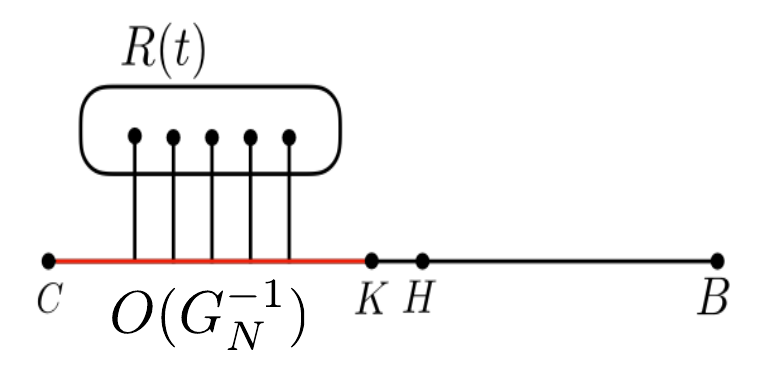}
         \caption{}
                 \label{fig:cartoon}
     \end{subfigure}
  \begin{subfigure}[b]{0.45\textwidth}
         \centering
         \includegraphics[width=0.3\textwidth]{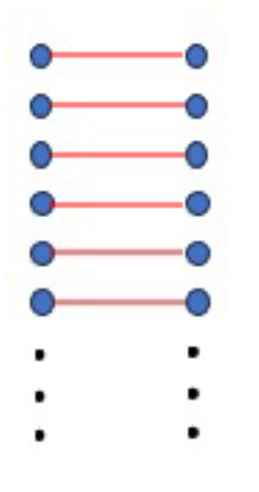}
         \caption{}
                 \label{fig:arakiwoods}
     \end{subfigure}
    \caption{(a) cartoon picture for the entanglement between the interior of the black hole and the radiation. A Cauchy slice is shown with $B$ denoting the boundary, $H$ the horizon, $K$ the QES that dominates after the Page time, and the spacetime smoothly capping off at $C$.
    (b) The example of Araki-Woods where two groups of $N \to \infty$ spins are pairwise entangled with each pair in a generic entangled state.}
    \label{fig:cart}
\end{figure}

At $t_1$, the black hole $B$ is entangled with the reservoir $R$ with $S_B \sim {\cal O}(1/G_{N})$, and the entanglement wedge $\sW_B$ for the black hole includes a complete Cauchy slice in the black hole geometry. See Fig.~\ref{fig:bulkwedges} and Fig.~\ref{fig:count}. 
The $O(1/G_N)$ entanglement comes solely from entanglement between matter excitations in the interior of the black hole and the reservoir\footnote{The AdS spacetime is coupled to the reservoir at the boundary -- there is an energy flux across there. 
Such couplings at the AdS boundary between the black hole and reservoir systems could lead to entanglement across the AdS boundary. Such entanglement, however, is not relevant for understanding the implications of the  entanglement between the interior of the black hole and the reservoir resulting from evaporation. To focus on the implications of the black hole evaporation on the state of the system at a given time $t$, we may simply decouple the two systems and evolve with the decoupled Hamiltonian. Such a decoupling will generate a shock wave, but it will not change the state of the system at time $t$.} 
with QES empty.
 In order to ``squeeze'' a large amount of entanglement into a Cauchy slice on which fluctuations are still suppressed with $G_N$, we expect the volume of the Cauchy slice to diverge with ${\cal O} (G_N^{-1})$. Indeed, this can be checked in the example of~\cite{AEMM}, and is expected to hold for generic evaporating black holes~\cite{SusZha14,BroGha19}. See Fig.~\ref{fig:cart}(a) for a cartoon.  
This immediately implies that the black hole spacetime at $t_1$ is quantum volatile rather than classical.

We had previously discussed the expectation that the algebra of bulk operators associated with a complete Cauchy slice is type I. A key assumption underpinning this conclusion is that no geometric quantity scales as $G_N^{a<0}$ in the limit that $G_N \to 0$. This assumption is violated by the black hole interior at $t\sim {\cal O}(G_{N}^{-1})$: the volume of the interior for most of this time range scales as $G_{N}^{-1}$. Consequently, the entanglement structure, illustrated schematically in Fig.~\ref{fig:cart}(a) for the black hole-reservoir system, resembles a famous example of Araki-Woods~\cite{AraWoo68} (see Fig.~\ref{fig:cart}(b)) where it is known that the algebra of operators is type III$_1$. Another related analogous system is a quantum field theory in flat space at finite temperature  in the infinite volume limit, 
which is also known to be generically type III$_1$ (see~\cite{Wit21a} for a review of the arguments). We thus expect that $\sA_{\sW_B}$ is type III$_1$, as is the boundary algebra for the black hole system $\sA_B$. It is interesting that a ``large'' amount of entanglement as a consequence of time evolution ``downgrades'' the corresponding algebras from type I to type III$_1$. This also results in a change in connectivity: per our proposal of Sec.~\ref{sec:prop}, $B$ and $R$ are quantum connected.

Consider now the full history of the evaporation of a black hole. At early times, the entanglement between $B$ and $R$ is of order $\sO (G_N^0)$, and so is the volume of the interior of the black hole, in which case we expect $\sA_B$ and $\sA_R$ to be type I. The situation is then similar to that of Fig.~\ref{fig:wh1}(b), and $B$ and $R$ are disconnected. 
At time $t_1< t_P$, when the entanglement between $B$ and $R$ and the volume of the interior are order $\sO (G_N^{-1})$, $\sA_B, \sA_R$ are type III$_{1}$: $B$ and $R$ are quantum connected. At $t_2 > t_P$,  $\sA_B$ and $\sA_R$ are both type III$_1$, and $\sW_B$  is classical. However, since $\sW_R$ is quantum volatile due to the island, application of our proposal is more subtle: we require the special provision of a spacetime with one classical entanglement wedge and one quantum volatile entanglement wedge as formulated in Sec.~\ref{sec:inbetween}. Recall that connectivity type in such scenarios is determined by the canonical purification of the classical of the two wedges. In this case, the prescription instructs us to canonically purify $B$, which generates a completely classical two-sided black hole (see Fig.~\ref{fig:cp} of Appendix~\ref{sec:grpu}). 
Thus $B$ and its canonical purification $\widetilde{B}$ is classically connected. 
This should be contrasted with the situation at $t_1$ where the canonical purification generates two copies of quantum volatile spacetime (see Fig.~\ref{fig:cp}). 
Thus during the evaporation of a black hole, there are three stages: at the beginning $B$ and $R$ are disconnected; then connected by a quantum wormhole; and finally connected by a classical wormhole.

\subsubsection{Spacetime fluctuations from crossed product} \label{sec:quan2}
Quantum volatility can also arise in the context of crossed product constructions, involving a microcanonical thermofield double state (at a sufficiently high energy) with energy fluctuations of $O(N^0)$~\cite{ChaPen22}. 
Such a state is still believed to describe a black hole system with two boundaries, but  the asymptotic time differences $t_L + t_R$ between two boundaries 
of the black hole geometry\footnote{We shoot a radial geodesic from the left boundary at time $t_L$, and $t_R$ is the right boundary time when the geodesic reaches the right boundary. Here $t_L, t_R$ are taken to go in the same direction.
In a classical spacetime, $t_L + t_R$ has no fluctuations in the $G_N \to 0$ limit.}
 have ${\cal O} (G_N^0)$ fluctuations. Furthermore, the  ``horizon'' area  has fluctuations of ${\cal O} (G_N)$, which means that the fluctuations of extrinsic curvature at the horizon are of order ${\cal O} (G_N^0)$. Here we put ``black hole'' and ``horizon'' in quotes as such concepts can now only be defined approximately in this context. 
 In this case, the boundary algebra for the right (left) system turns out to be type~II from the crossed product of a type III$_1$ algebra of single-trace operators with its modular flow~\cite{Wit21b, ChaPen22}.  The right and left systems are now quantum connected. 
 
We can view the emergence of the type II structure as a result of bulk quantum fluctuations ``upgrading'' the type III$_1$ algebra to type II. 
 Conversely, we can also say that bulk quantum fluctuations are bulk reflections of type II structure of the boundary operator algebras.  
More explicitly, consider a type III$_1$ von Neumann algebra $\sA$ in a cyclic and separating state, with the corresponding modular action denoted as 
$\alpha_t (a) = e^{i K t} a e^{- i Kt}$, where $K$ is the modular Hamiltonian, $t \in \RR$ is the modular time,  and $a \in \sA$. 
The crossed product $\sA \otimes_\alpha \RR$ acts on the Hilbert space $ {\cal K} = {\cal H} \otimes L^2 (\RR)$, where ${\cal H}$ is the Hilbert space acted on by $\sA$. In the example of microcanonical TFD, $\sA$ is the single-trace operator algebra of the right CFT (or the  operator algebra of bulk quantum fields in the right region of the black hole), and $\sH$ is the Fock space of quantum matter fields in the black hole spacetime. 
The crossed product plays the role of ``splitting'' the modular time into independent times $t_R$ and $t_L$ for $\sA$ and $\sA'$ respectively, and $L^2 (\RR)$ can be interpreted as the space of $L^2$ functions on the space $\RR$ of a new relative time $p = t_R + t_L$\footnote{Modular flow acts on $\sA$ and $\sA'$ in opposite manner, i.e. it leaves $t_R + t_L$ invariant.} 
The quantum mechanics of $L^2 (\RR)$ has $\hbar \sim {\cal O} (G_N^0)$. 
The physical state in $\sK$ 
involves a normalizable wave function $\psi (p)$ on $L^2 (\RR)$, in which the relative time $t_L + t_R$ always has fluctuations of ${\cal O} (G_N^0)$.\footnote{A state with no fluctuation must be proportional to $\delta (p)$, which is not normalizable.}

 A closely related example is the de Sitter spacetime in the presence of a pair of bulk observers 
 in complementary static patches~\cite{ChaLon22}. 
 In the joint state of the observers and quantum matter excitations (including the metric) in de Sitter, again $t_R + t_L$ has  order ${\cal O} (G_N^0)$ fluctuations in the $G_N \to 0$ limit, and the operator algebra associated with a static patch become type~II.  
The discussion of~\cite{Wit21b, ChaLon22, ChaPen22} has been generalized to general regions (both in AdS and more generally)~\cite{JenSor23} and other black holes~\cite{KudLeu23}, and similar statements can be made there.

Type II algebras have also appeared in the discussion of JT gravity with a finite $G_N$~\cite{PenWit23} and double scaled SYK 
models~\cite{Lin22}. In both of these examples, the bulk spacetime fluctuations are unsuppressed. They do not strictly fall under the purview of our classification as our regime is restricted to the $G_{N}\rightarrow 0$ limit. Nevertheless, these examples illustrates that at least some of our notions of connectivity can in principle be generalized to the finite $G_{N}$ regime, though we will not do so here.

To summarize, there are close ties between the crossed product structure with the associated 
type II algebras and quantum volatile spacetimes, which reflects that the corresponding states have very different entanglement structure from those describing classical spacetimes.

\subsubsection{Time-dependent algebraic ER=EPR} \label{sec:quan3}

Can connectivity change dynamically? Our current definition of a classical spacetime is given in terms of the behavior of a family of time slices spanning $O(G_N^0)$ of boundary times. Two asymptotic boundaries are connected if there exists \textit{some} spacelike path, anchored within this time span on the boundary, that connects the two boundaries through the bulk within the family. This definition admits a natural generalization to allow for a dynamical change in the spacetime from classical to quantum volatile, and a subsequent dynamical change in connectivity from classical to quantum or vice versa. 

We may simply say that a spacetime is classical in some open set ${\cal U}$ if ${\cal U}$ intersects $\partial M$ on some ${\cal O}(G_{N}^{0})$ time band, and ${\cal U}$ satisfies our definition of a classical spacetime. We can similarly say that a spacetime is quantum volatile in some open set ${\cal U}$ as above if ${\cal U}$ satisfies the definition of a quantum volatile spacetime. Under this generalization, we could have a dynamical change in connectivity as we can consider different families of Cauchy slices anchored at different boundary times.

 We have already seen this shift in the example of an evaporating black hole: $B$ and $R$ change from being disconnected (stage 1) to quantum connected (stage 2) to classically connected~(stage 3). $\sA_B$ and $\sA_R$ change from type I at stage 1 to type III$_1$ at stages 2 and 3. 
The time scale separating stage 1 and stage 2 is of order $O(G_N^{-1})$ while 
that separating between stage 2 and stage 3 can be of order $O(G_N^{0})$. 
It is intuitively reasonable that spacetime connectivity or the type of boundary algebras can change over the time scale of $O(G_N^{a<0})$. It is more surprising that $B$ and $R$ can change from quantum disconnected to classically disconnected over the time scale of order $O(G_N^0)$. We will elaborate this further in Sec.~\ref{sec:trans} where we identify a microscopic mechanism responsible for this and argue that this is a consequence of a dramatic change of the content of the boundary algebras, although not of the algebra type.  

Here we point out another simple example where a dynamical change in connectivity happens. Consider again the two CFTs in a thermofield double state, dual to the Schwarzschild-AdS black hole. This spacetime was used as a prototypical example in the introduction for classical connectivity under the strict  definition of Sec.~\ref{sec:prop} that did not allow for dynamical changes. Under the relaxed definition that does allow for changes in classicality of the spacetime, we may also consider a late-time Cauchy slice anchored at $t_{L},t_{R}\sim {\cal O}(1/G_{N})$. The vacuum-subtracted volume of the slice will diverge as $G_{N}\rightarrow 0$. At such times we then have a quantum volatile spacetime;
the boundary algebras at such times are likely still type III$_1$.\footnote{Defining the large $N$ limit of the boundary algebras at such times is subtle; we will not go into the details here.} Thus the spacetime shifts from quantum connectivity at early times ($t_{L},t_{R}\sim- {\cal O}(1/G_{N})$) to classical connectivity and back to quantum quantum connectivity at late times.

\begin{figure}
    \centering
    \includegraphics[width=0.5\textwidth]{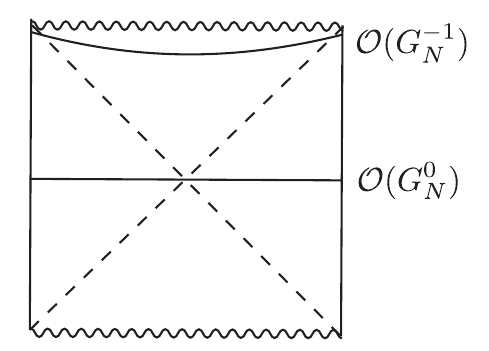}
    \caption{Two Cauchy slices of the AdS Schwarzschild black hole: at intermediate times, the spacetime is classical while at late times the spacetime is quantum volatile. }
    \label{fig:LateTimes}
\end{figure}

\section{Connectivity for Multipartite States}\label{sec:multi}

We proceed to generalize the above proposal for bipartite states to multipartite states. Now $R_1, R_2$ may be  
two boundaries of an $n$-boundary system, or $R_1$ and/or $R_2$ may be subregions of a given asymptotic boundary or of multiple asymptotic boundaries. We will also have occasion in Sec.~\ref{sec:aev} and in Sec.~\ref{sec:trans} to discuss algebras of operators that are \textit{not} associated to a geometric decomposition of the CFT Hilbert space. 
An important tool for our discussion is the gravitational construction~\cite{EngWal18} of canonical purification, which we now review. 

\subsection{Canonical Purification and its Gravitational Dual} \label{sec:cano}

Given a mixed state (a density matrix) $\psi$ acting on some Hilbert space ${\cal H}$, 
with the decomposition in 
its eigenbasis:
\be\label{dep}
\psi = \sum \limits_{i}\lambda_{i}\ket{i}\bra{i},
\ee
the canonical purification of $\psi$, denoted $\ket{\sqrt{\psi}}$ is a pure state on a doubled Hilbert space ${\cal H}\otimes {\cal H}$, defined by  ($\ket{\widetilde i}$ is the time reversal of $\ket{i}$)
\be\label{cp1}
\ket{\sqrt{\psi}} =  \sum \limits_{i}\sqrt{\lambda_{i}}\ket{i}\ket{\widetilde i} \ .
\ee
This definition may be thought of as a generalization for arbitrary mixed states of the map from the Gibbs ensemble to the thermofield double state. 

In many physical applications, including the evaporating black hole example discussed earlier, 
we are interested in a subsystem $R$ with
an associated operator algebra $\sA_R$. If $\sA_R$ is type I, then there is a local Hilbert space $\sH_R$ associated with $R$ and the canonical purification procedure described above can be applied to a density operator $\psi$ on $\sH_R$. However, in general in the $G_N \to 0$ limit, $\sA_R$ can become type II or III, in which case there does not exist a Hilbert space associated with $R$, and a density operator $\psi$ in the form of~\eqref{dep} cannot be defined.\footnote{For a type II algebra, a ``renormalized'' density operator can in principle defined, but it does not act on a local Hilbert space
and cannot be defined as~\eqref{dep}.} Thus the procedure~\eqref{cp1} cannot be used.   

This difficulty can be circumvented as follows.  A state $\psi$ on an algebra $\sA_R$ can be defined as a map from $\sA_R$ to $\field{C}$, i.e. 
$\psi (a) \in \field{C}, a \in \sA_R$. A GNS Hilbert space $\sH_\psi$ can be built around $\psi$ by
associating a vector $\ket{a} \in \sH_\psi$ to an element $a \in \sA_R$.\footnote{Heuristically we may interpret $\ket{a}$ as obtained by acting $a$ on $\psi$.} The state $\ket{\bid}\in \sH_\psi$ corresponding to the identity operator can be viewed as the canonical purification of state $\psi$. When $\sA_R$ is type I, this procedure leads to the same result as~\eqref{cp1}.

\begin{figure}[t]
    \centering
   \includegraphics[width=0.8\textwidth]{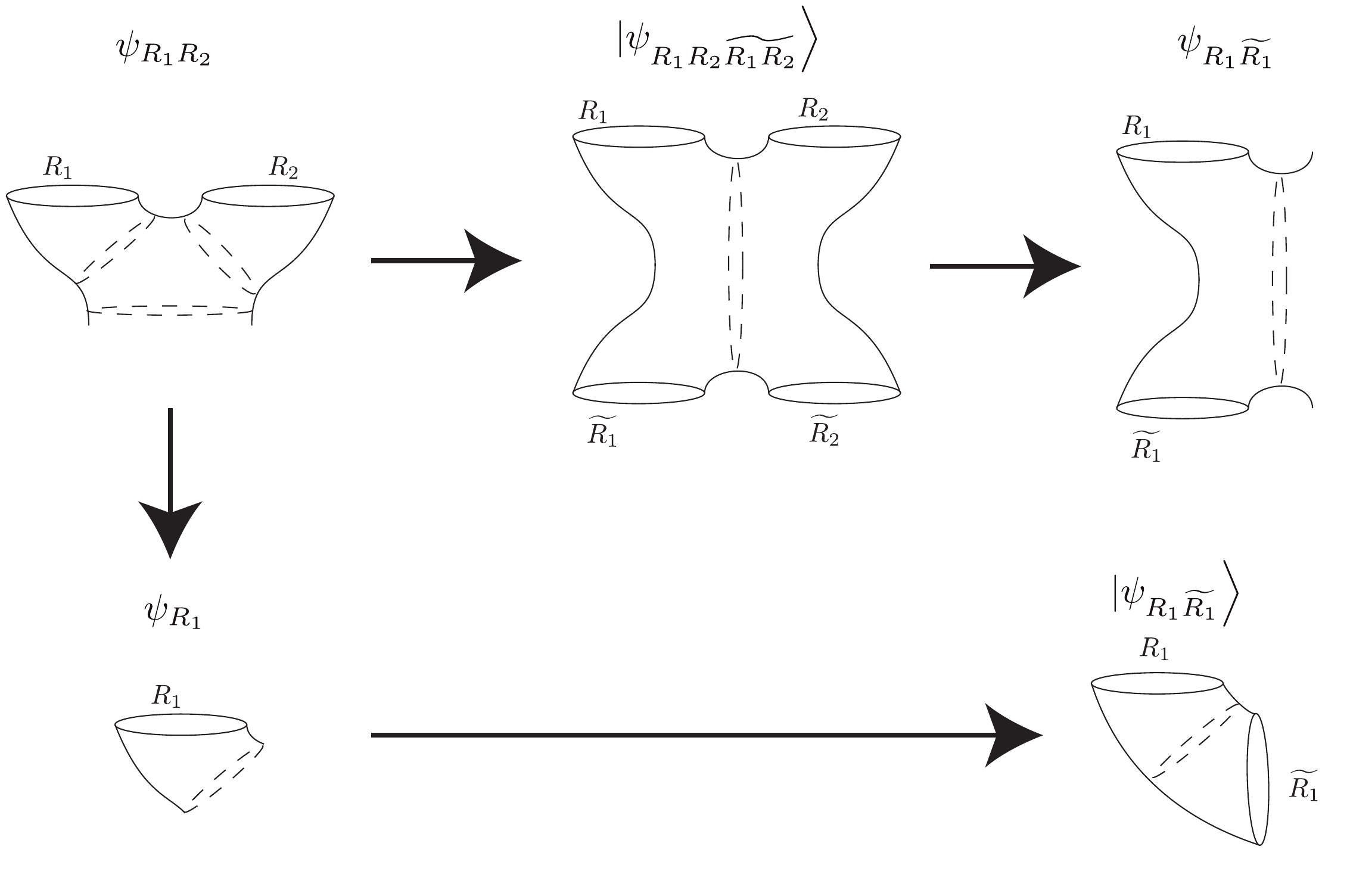}
    \caption{Examples of the gravitational dual of the canonical purification. 
    Top left: the entanglement wedge of a mixed state $\psi_{R_{1}R_{2}}$ on a Cauchy slice, with dashed lines marking the QESs of $\psi_{R_{1}}$, $\psi_{R_{2}}$ and $\psi_{R_{1}R_{2}}$. Top middle: the canonical purification of $\psi_{R_{1}R_{2}}$ into $\ket{\psi_{R_{1}R_{2}\widetilde{R_{1}R_{2}}}}$ on a Cauchy slice by CPT conjugating around the QES of $\psi_{R_{1}R_{2}}$; the dashed line is the QES of $\psi_{R_{1}\widetilde{R_{1}}}$. Top right: the entanglement wedge of $R_{1}\widetilde{R_{1}}$ in the state  $\ket{\psi_{R_{1}R_{2}\widetilde{R_{1}R_{2}}}}$. Bottom left: the state $\psi_{R_{1}}$ obtained by reducing $\psi_{R_{1}R_{2}}$ on $R_{2}$. Bottom right: the canonical purification of $\psi_{R_{1}}$. 
    }
   \label{fig:cpex}
\end{figure}

Now consider a system with a classical or quantum volatile bulk dual, and let $R$ be a boundary region, either be  edgeless or with edge, connected or disconnected. 
Suppose that $\sW_R$ is the entanglement wedge of $R$, with the corresponding QES denoted as $\chi$. From subalgebra-subregion duality, the bulk operator algebra associated with $\sW_R$ is identified with the boundary algebra $\sA_R$ associated with $R$.\footnote{$\sA_R$ is defined by taking appropriate large $N$ limit and contains operators generated by modular flow~\cite{LeuLiu22}.} The geometry of $\sW_R$ defines a state $\psi$ on $\sA_R$. 

The gravitational dual of the canonical purification of $\psi$ is obtained by CPT-conjugating the entanglement wedge around its QES $\chi$ and gluing the CPT-conjugated wedge to the original wedge across $\chi$~\cite{EngWal18}. See~\cite{DutFau19} for a path integral justification of this protocol. 
See Fig.~\ref{fig:cpex} for some examples and subtleties. In general, the construction leads to a spacetime whose asymptotic boundaries are given by two copies of $R$.   We will denote the new copy as $\widetilde R$. The state corresponding to the resulting spacetime 
is by definition a pure state on the system $R \widetilde R$.

\subsection{Formulation} 

We now proceed to discuss multipartite states. In such cases, we shall need to invoke the canonical purification because a straightforward application of the proposal of the last section immediately breaks down, as illustrated by the following examples: 

\ben 

\begin{figure}
    \centering
    \includegraphics[width=0.8\textwidth]{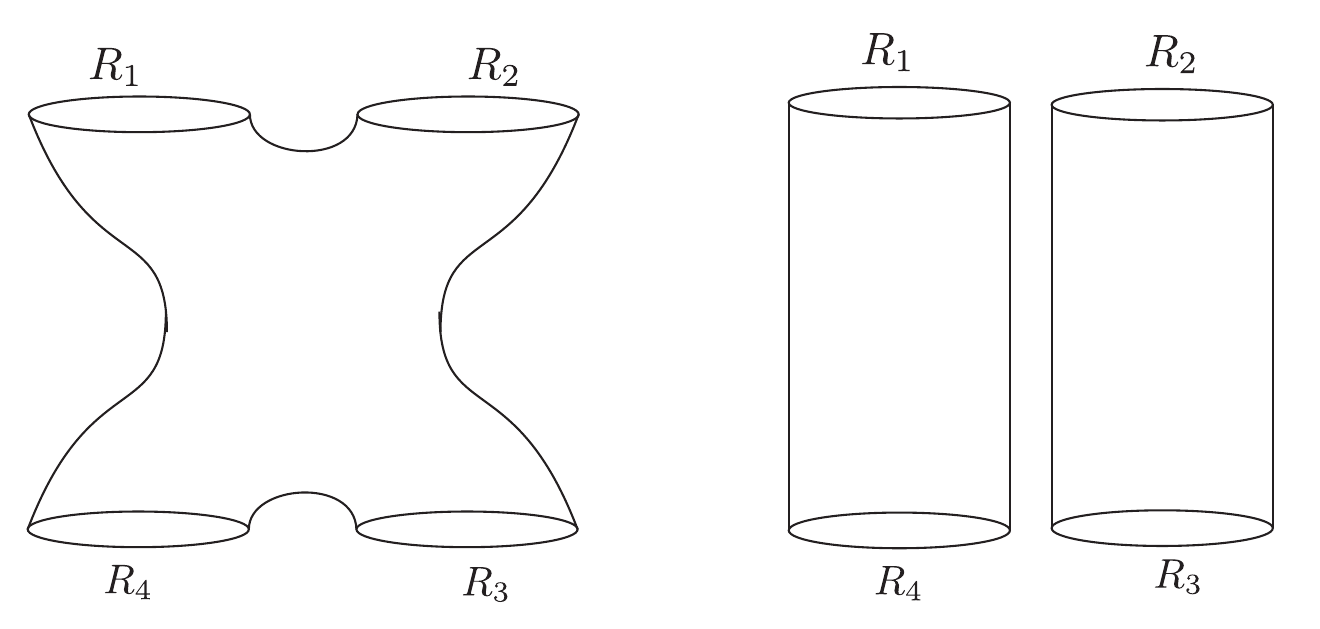}
    \caption{Two examples of Cauchy slices of four-boundary spacetimes: on the left, the spacetime is full connected, and on the right it has two disconnected components. The algebras of operators associated to $R_{1}$ and to $R_{2}$ are type III$_{1}$ in both Cauchy slices above, even though $R_{1}$ and $R_{2}$ are connected in the left figure and disconnected in the right.}
    \label{fig:R1R2R3R4}
\end{figure}

\item 
Let $\ket{\psi_{R_{1}R_{2}R_{3}R_{4}}}$ be a 4-party state satisfying the classical condition where each of the $R_{i}$ is a complete $\mathscr{I}$, and
the ${\cal A}_{R_{i}}$ are all type III$_{1}$. By our arguments from the previous section, the $\sW_{R_{i}}$ all have boundaries in the interior of the spacetime. However, this boundary may be shared (on some Cauchy slice) with just one other $\sW_{R_{i}}$ or with all the remaining three. In particular,  just knowing that ${\cal A}_{R_{1}}$ and ${\cal A}_{R_{2}}$ are both type III$_{1}$ is insufficient to deduce that $\sW_{R_{1}\cup R_{2}}$ connects $R_{1}$ to $R_{2}$, since it is possible that $R_{1}$ is in fact connected to $R_{3}$ and $R_{2}$ is connected to $R_{4}$. See Fig.~\ref{fig:R1R2R3R4} for an illustration.

\item Consider now a proper subregion of $\mathscr{I}$. By definition, its entanglement wedge is also a subregion, and thus the associated algebra will always be type III$_1$ in the $G_N \to 0$ limit, irrespective of connectivity properties.\footnote{Type III$_1$ here is emergent in the $G_N \to 0$ limit, and should be distinguished from type III$_1$ algebra of a subregion at finite $N$. We can imagine always putting the boundary theory on lattice so that all subalgebras associated with a subregion are type I at finite $N$.}
 An immediate example of this is given by two intervals on $S_1$ in a two-dimensional CFT 
in the vacuum state: whether the entanglement wedge is connected or not, the algebra associated to the state on each subregion will be type III$_{1}$ and the state will satisfy the classical condition. 

\een

It turns out that there is a simple unified way to address these complications: we isolate the regions of interest and purify the state via the canonical purification construction reviewed above. 

More explicitly, consider the example of item 1. To illustrate the logic with clarity, we will temporarily assume that the state in question satisfies the classical condition; we will drop this assumption prior to stating our full proposal. Isolate the subsystem $R_1 \cup R_2$ and consider its canonical purification $\ket{\psi_{R_{1}R_{2}\widetilde{R_{1}R_{2}}}}$ with corresponding gravitational dual $\sW_{R_{1}R_{2}\widetilde{R_{1}R_{2}}}$.  If ${\cal A}_{R_{1}}$ and ${\cal A}_{R_{2}}$ are both type III$_{1}$ factors, we have -- even after we have assumed that the bulk volumes, areas, and lengths do not scale with $G_{N}^{a<0}$ -- two possibilities for $\sW_{R_{1}R_{2}\widetilde{R_{1}R_{2}}}$: either $R_{1}, R_{2}, \widetilde{R_{1}}$, and $\widetilde{R_{2}}$ are \textit{all} connected to one another (i.e. a single connected four-boundary spacetime) or $R_{1}$ is connected to $\widetilde{R_{1}}$ but not to $R_{2}$ and $\widetilde{R_{2}}$, and $R_{2}$ is connected to $\widetilde{R_{2}}$ \footnote{Note that $R_{1}$ cannot be connected to $\widetilde{R_{2}}$ and not to $R_{2}$ since the CPT gluing must be done at the boundary of the entanglement wedge.} In the former case, $\sW_{R_{1}}$ and $\sW_{R_{2}}$ share part of a boundary in the bulk interior; in the latter case, they do not. In the former case, 
the entanglement wedge $\sW_{R_1 \cup R_2}$ of $R_1 \cup R_2$ connects $R_{1}$ to $R_{2}$, while in the latter case it does not. See Fig.~\ref{fig:MultiCPT}. 
The distinguishing boundary feature between the two cases is the algebra ${\cal A}_{R_{1}\widetilde{R_{1}}}$ associated with the canonical purification $\sW_{R_{1}R_{2}\widetilde{R_{1}R_{2}}}$ (not from the canonical purification of $R_1$).\footnote{{There is a potential subtlety in the canonical purification of subregions: if there is a cusp at the gluing, operators localized to the cusp may be ill-defined. We expect that there is some limiting procedure that regulates this issue; in the absence of such a definition, we may choose to implement some lattice spacing or restrict to regions which glue onto their CPT conjugates in a differentiable fashion.}} When $\sW_{R_1 \cup R_2}$  connects $R_{1}$ to $R_{2}$, ${\cal A}_{R_{1}\widetilde{R_{1}}}$ is type III$_{1}$, since there must be an interior spacetime boundary onto which $\sW_{R_{2}}$ can connect to $\sW_{R_1 \cup \widetilde{R_{1}}}$. When ${\cal A}_{R_{1}\widetilde{R_{1}}}$ is type I, $\sW_{R_1 \cup R_2}$ fails to connect $R_{1}$ to $R_{2}$.

\begin{figure}
    \centering
    \includegraphics[width=\textwidth]{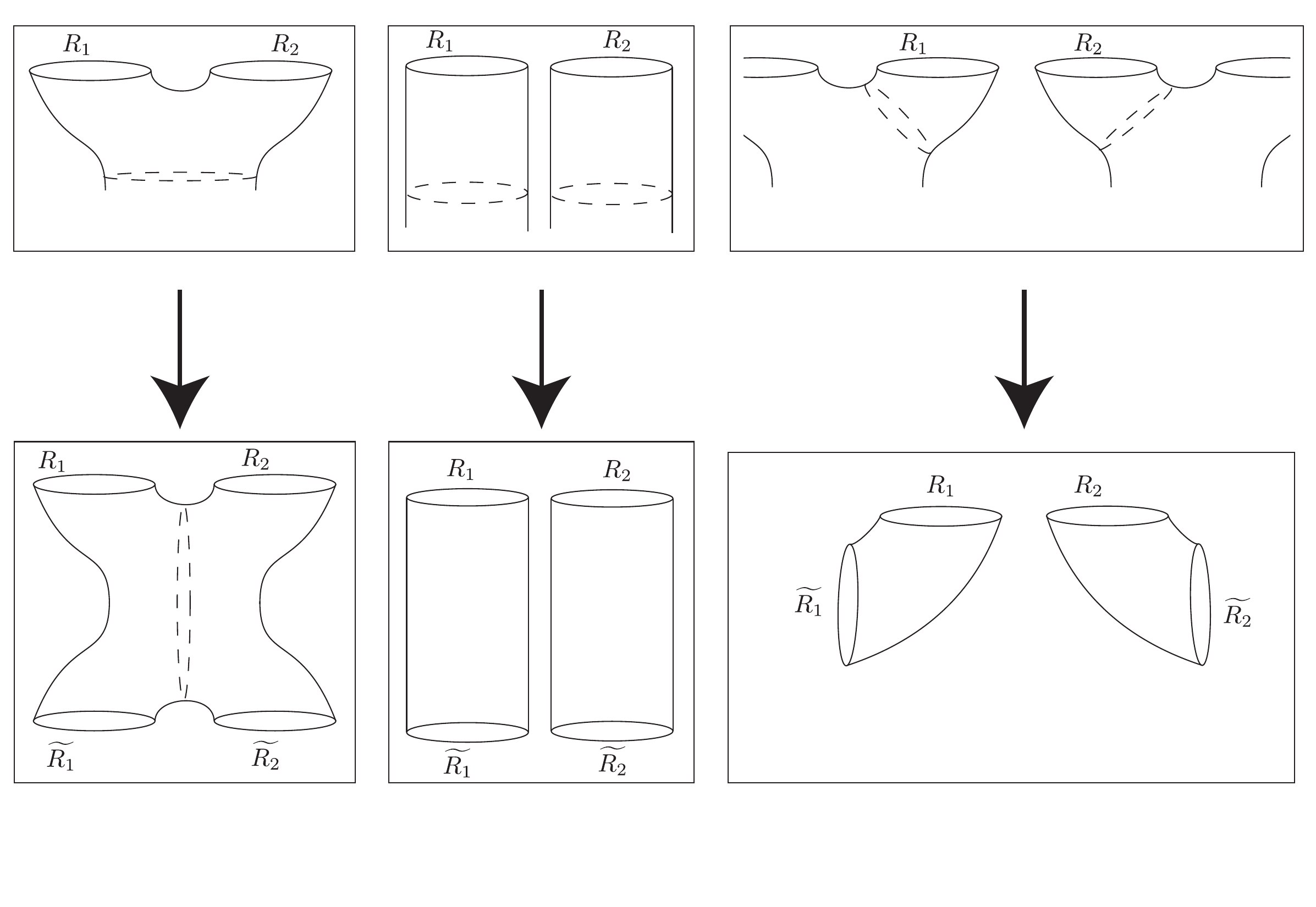}
    \caption{The gravitational dual of the canonical purification of $\psi_{R_{1}R_{2}}$ in several different setups. On the left, $\psi_{R_{1}R_{2}}$ is dual to a connected wedge, and the algebra ${\cal A}_{R_{1}\widetilde{R_{1}}}$ has a bulk interior boundary and is therefore type III$_{1}$. In the middle, $\psi_{R_{1}R_{2}}$ is dual to a disconnected wedge, and each $\psi_{R_{i}}$ is in fact in a bipartite state with another system. ${\cal A}_{R_{1}\widetilde{R_{1}}}$ is not type III$_{1}$. On the right, $R_{1}$ and $R_{2}$ are dual to disconnected wedges and are not in a bipartite state with another system, but the canonical purification still serves as a diagnostic: ${\cal A}_{R_{1}\widetilde{R_{1}}}$ is note type III$_{1}$.}
    \label{fig:MultiCPT}
\end{figure}

The above considerations then motivate a proposal entirely analogous to the bipartite case by ``bipartiting'' the multipartite case.

\paragraph{Proposal for General (Holographic) States:} 
Let $\sW_{R_1 \cup R_2}$ be the entanglement wedge dual to $\psi_{R_1 \cup R_2}$, and $\sW_{R_{1}R_{2}\widetilde{R_{1}R_{2}}}$
be the gravitational dual to the canonical purification of $\psi_{R_1 \cup R_2}$. $\sA_{R_i \widetilde{R_{i}}}$ with $i=1,2$ denotes the algebra for $R_i \cup \widetilde{R_{i}}$ resulting from the purification.\footnote{It is critical here that $A_{R_{i}\widetilde{R_{i}}}$ is \textit{not} obtained from the canonical purification of $R_i$  but rather from canonically purifying $R_{1}\cup R_{2}$.}
We assume that $\sW_{R_1 \cup R_2}$ and $\sW_{R_{1}R_{2}\widetilde{R_{1}R_{2}}}$ are semiclassical, i.e. classical or quantum volatile. 

\textit{\begin{enumerate}[I.]
\item $\sW_{R_1 \cup R_2}$ fails to connect $R_{1}$ to $R_{2}$ if 
	 ${\cal A}_{R_{1}\widetilde{R_{1}}}$ and ${\cal A}_{R_{2}\widetilde{R_{2}}}$ are both type I factors. 
	\item  $\sW_{R_1 \cup R_2}$ classically connects $R_{1}$ to $R_{2}$ if and only if  ${\cal A}_{R_{1}\widetilde{R_{1}}}$ and ${\cal A}_{R_{2}\widetilde{R_{2}}}$ are both type III$_{1}$ factors and furthermore the classical condition holds in $\psi_{R_{1}R_{2}}$.
	\item $\sW_{R_1 \cup R_2}$ quantum connects $R_{1}$ to $R_{2}$  if ${\cal A}_{R_{1}\widetilde{R_{1}}}$ and ${\cal A}_{R_{2}\widetilde{R_{2}}}$ are not type I and the classical 
	 condition is violated in $\psi_{R_{1}R_{2}}$.	
\end{enumerate}}\

The reasoning for this proposal follows from the definition of the gravitational canonical purification protocol discussed in Appendix~\ref{sec:grpu}:
$\sW_{R_{1}R_{2}\widetilde{R_{1}R_{2}}}$ describes a  purified state  $\ket{\psi_{R_{1}R_{2}\widetilde{R_{1}R_{2}}}}$ that can be thought of as a bipartite state on two systems: $R_{1}\cup\widetilde{R_{1}}=B_{1}$ and $R_{2}\cup\widetilde{R_{2}}=B_{2}$. 
When $B_{1}$ and $B_{2}$ are complete $\mathscr{I}$'s, this reduces precisely to the bipartite case of Sec.~\ref{sec:prop}. 
The only exception (to $B_1, B_2$ being edgeless) is when $R_1$ and $R_2$ are connected themselves.
In this case $B_1$ and $B_2$ are subregions 
on a single boundary, and $\sA_{B_i}$ are both type III$_1$ and the finiteness condition holds, so the proposal is trivially satisfied.

\subsection{Connectivity between the island and Hawking radiation} \label{sec:aev}

As an example of application of multipartite ER=EPR we revisit the evaporating black hole and consider the connectivity between the island and the {\it Hawking} part of the reservoir, which is what we shall call the subset of the reservoir ignorant of the black hole interior (see e.g.~\cite{AkeEng22}).

Consider a time after the Page time $t > t_P$, when the QES for $B$ lies slightly behind the event horizon. 
As is evident from the bulk, the fundamental description of $R$ naturally separates into two parts~\cite{AEMM, Pen19, AlmMah19a},  
\be \label{ennp}
R = R_{\rm Hawk}\cup I 
\ee
where $R_{\rm Hawk}$ denotes the semiclassical part of the reservoir defined through tracking Hawking radiation from the black hole at the semiclassical level (and lives in a half-Minkowski spacetime), and $I$ denotes the part (island) lying behind the QES in the interior of the black hole, which is in the entanglement wedge of the reservoir. See Fig.~\ref{fig:IslHawk}. Because the bulk domains of dependence are disjoint, the algebras are separate as well: denoting the algebras associated with $R_{\rm Hawk}$ and $I$ respectively as $\sA_{R_{\rm Hawk}}$ and $\sA_{I}$, we can write~\eqref{ennp} algebraically as 
\be \label{yer0} 
\sA_R = \sA_{R_{\rm Hawk}} \lor \sA_{I}, \quad t > t_P \ .
\ee
By contrast, before the Page time we simply have 
\be\label{yer}
\sA_R = \sA_{R_{\rm Hawk}}, \quad t < t_P \ .
\ee

\begin{figure}
    \centering
    \includegraphics{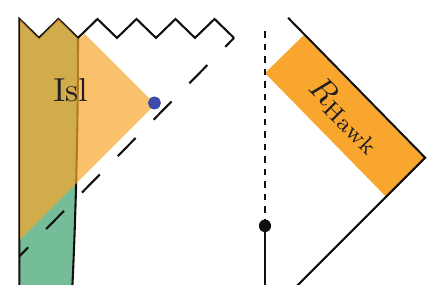}
    \caption{The decomposition of $R$ into the Hawking part and the island.}
    \label{fig:IslHawk}
\end{figure}

$\sA_{I}$ is type III$_1$ as it is the algebra associated with a bulk subregion. $I$ is connected to $B$ classically through the QES. {(This follows from the canonical purification of $B$, even though $I$ is quantum volatile.)} We may ask how and whether 
the island is connected to $R_{\rm Hawk}$. Since a pure state is defined only for the union $R_{\rm Hawk} \cup I \cup B$, this question 
concerns a multipartite state.  Intuitively, given (a) the $\sO(G_N^{-1})$ entanglement between $R_{\rm Hawk}$ and $I$, (b) the predictions originally expressed in~\cite{Van13, MalSus13} as a way of relaxing the firewall problem, and (c) the fact that $R_{\rm Hawk}$ and $I$ are \textit{not} classically connected, we expect that they are connected by a quantum wormhole. We will show that this is indeed the case following the multipartite algebraic ER=EPR proposal.

According to the proposal, we need to construct the gravitational canonical purification $\sW_{R_{\rm Hawk} I  \widetilde{R_{\rm Hawk} I}}$ for $R = R_{\rm Hawk} \cup I$. $\sW_{R_{\rm Hawk} I  \widetilde{R_{\rm Hawk} I}}$ is built via CPT conjugation around the minimal QES -- the QES given by the union of the boundary of the island and the AdS boundary. The result is a semiclassically disconnected spacetime: Minkowski spacetime (with stress tensor shocks) and a baby universe (BU), see Fig.~\ref{fig:BabyUniverse},  with the corresponding algebras $\sA_{\rm BU} 
= \sA_{I \widetilde I}$ and $\sA_{\rm Mink} = \sA_{R_{\rm Hawk} \widetilde {R_{\rm Hawk}}}$. 
 The baby universe and the resulting Minkowski spacetime each contain a inextendible Cauchy slice, but the volume of a Cauchy slice of the baby universe is divergent in $G_{N}^{-1}$: we expect that the algebras ${\cal A}_{BU}$ and 
 $\sA_{\rm Mink}$ are type III$_{1}$. It then follows from multipartite algebraic ER=EPR that the island has a quantum wormhole to the radiation.

\begin{figure}
    \centering
\includegraphics[width=0.8\textwidth]{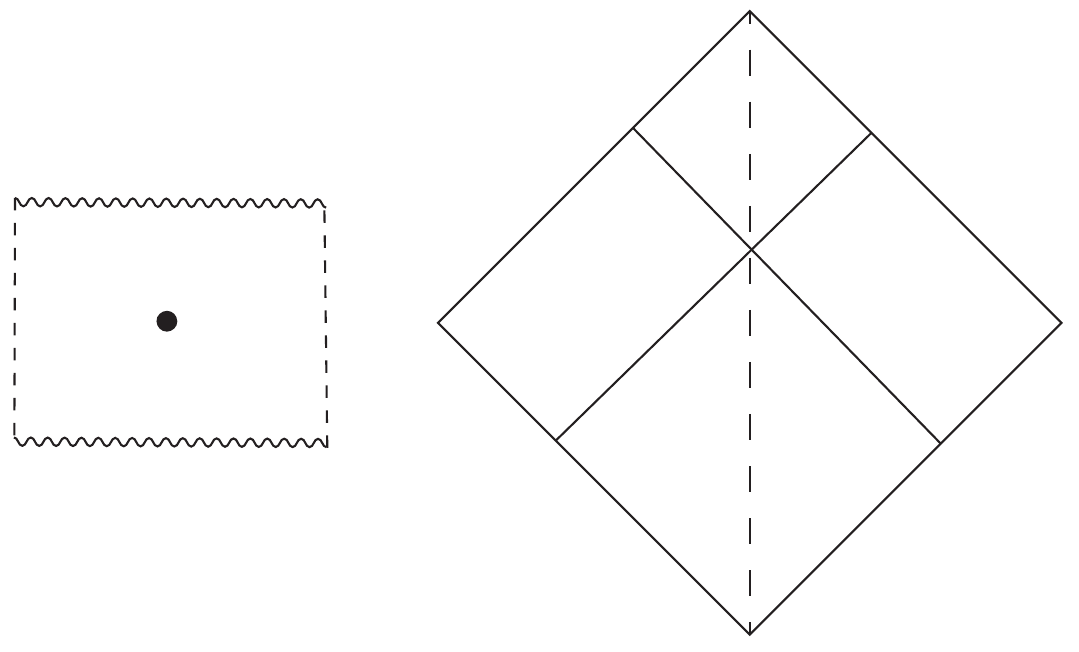}
    \caption{The canonical purification of the reservoir (obtained by decoupling from $B$ and evolving with the decoupled Hamiltonian of $R\cup \widetilde{R}$) after the Page time has two components connected by a quantum wormhole: on the left, the baby universe obtained via CPT conjugation of the island (the black dot is the QES of each side); on the right, the Minkowski region. }
    \label{fig:BabyUniverse}
\end{figure}
 
 So far our discussion has relied on the bulk picture, where there is a clean separation between $R_{\rm Hawk}$ and $I$. To have a genuine boundary description of the story, however, we should also give an intrinsic boundary definition of $R_{\rm Hawk}$ and $I$ as parts of $R$. 
 This can be done using the concept of complexity using ideas from~\cite{HarHay13, EngWal17b, EngWal18, BroGha19, EngPen21a, EngPen21b, AkeEng22}, which we will discuss in detail in Sec.~\ref{sec:trans}. 
 The basic picture is that  ${\cal A}_{I}$ is the algebra of complex operators in $R$ while $\sA_{R_{\rm Hawk}}$ consists of the algebra of simple operators in $R$ in the large-$N$ limit. Both of them are type III$_1$, but neither of them is defined geometrically on the boundary. So here the quantum wormhole is between the simple and complex sectors of the reservoir! This example thus illustrates the power of the algebraic formalism in capturing the bulk physics.

\section{Transition from Complexity Transfer} \label{sec:trans} 

We now turn to the question of what happens at the Page time posed in Sec.~\ref{sec:intro}. 
From the discussion of the previous section, we see that at the Page time there is a transition from $B$ and $R$ being quantum connected  to being classically connected. We would like to understand whether there exists a more fundamental way to think about the emergence of various properties at the Page time. In this section we will argue that the ``transition'' at the Page time can be interpreted as an algebraic realization of Harlow-Hayden: a transfer of complexity between $B$ and $R$ systems. In fact, this complexity transfer is a general phenomenon occurring in many other contexts, specifically when there is a switchover in dominance between QESs. 

\subsection{The Complex and Simple Factors}\label{sec:complex}

\begin{figure}
    \centering
    \includegraphics{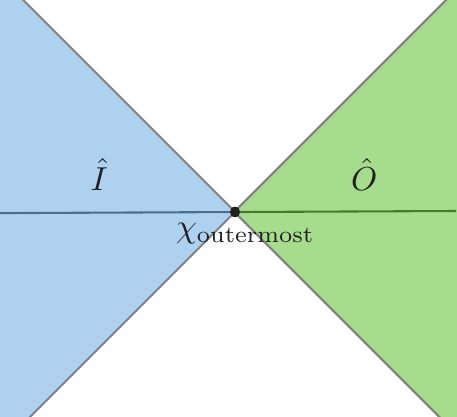}
    \caption{The decomposition of the entanglement wedge into $\hat{O}$ and $\hat{I}$, the outer wedge (simply reconstructible) and the interior of $\chi_{\rm outermost}$, which has high reconstruction complexity.}
    \label{fig:IOdecomp}
\end{figure}

Consider a boundary subsystem $Q$, which can be edgeless or with edge, whose entanglement wedge is $\sW_Q$. 
In addition to the minimal QES $\chi_Q$ that defines $\sW_Q$, $\sW_Q$ may contain other non-minimal QESs.  
We denote the one that is closest to the boundary (i.e. is contained in the wedge of all other QESs of $Q$) as $\chi_{\rm outermost}$\footnote{An outermost QES always exists~\cite{EngPen21a, EngPen23}, though it may be the empty set.}. 
It can be shown~\cite{EngPen21a, EngPen23} using maximin techniques~\cite{Wal12, AkeEng19b} that $\chi_{\rm outermost}$ is locally minimal on some Cauchy slice, and that furthemore between $\chi_R$ and $\chi_{\rm outermost}$ there must exist a QES $\chi_{\rm bulge}$ which is locally maximal on a Cauchy slice.\footnote{There are some subtleties in the presence of multiple such bulges; see~\cite{EngPen21b, EngPen23} for a discussion. } 
Let $\hat{O}$ be the wedge defined by the QES $\chi_{\rm outermost}$ and the boundary region for which it is a QES, and {let $\hat{I}$ be the domain of dependence between $\chi_{\rm outermost}$ and $\chi_{Q}$.} See Fig.~\ref{fig:IOdecomp}. The former was termed the simple wedge by~\cite{EngWal17a, EngWal18};  the latter is the so-called Python's Lunch, due to its shape on the Cauchy slice on which $\chi_{\rm outermost}$ is minimal and $\chi_{\rm bulge}$ is maximal. See Fig.~\ref{fig:PL}

\begin{figure}
    \centering
    \includegraphics[width=0.8\textwidth]{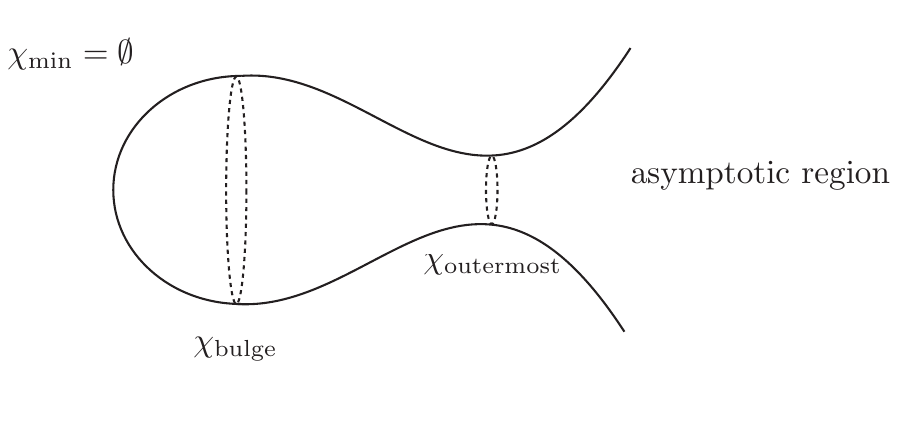}
    \caption{The outermost QES (sometimes called the throat or appetizer) and the bulge of a Python's Lunch on a Cauchy slice of a single-sided geometry.}
    \label{fig:PL}
\end{figure}

According to the strong Python's Lunch proposal~\cite{BroGha19,EngPen21a, EngPen21b}, operators in the simple wedge $\hat O$ 
can be reconstructed with subexponential complexity in the size of the code subspace on which the reconstruction is taking place, often taken to be $S$, the entropy of the black hole in question. For our purposes here, the relevant parameter is $G_{N}$: when we say that reconstruction complexity is subexponential, we mean that its scaling with $G_{N}$ is subexponential. Operators in the Python's Lunch, by contrast, are hard to reconstruct, with reconstruction complexity scaling as
\be\label{cx}
 \sC \propto \exp \left[\frac{S_{\rm gen}[\chi_{\rm bulge}]- S_{\rm gen}[\chi_{\rm outermost}]}{2} \right]
 \ .
\ee
Generically the bulge and outermost QESs differ in $S_{\rm gen}$ at leading order, so the reconstruction complexity of the lunch is generically exponential.

Denote the bulk operator algebras in regions $\hat O, \hat I$ and $\sW_Q$ respectively as $\sA_O, \sA_I$ and $\sA_{\sW_Q}$. 
We have 
\be \label{des1}
\sA_{\sW_Q} = \sA_O \lor \sA_I \ .
\ee
Since $O$ and $I$ are bulk subregions, both $\sA_O$ and $\sA_I$ are type III$_1$ factors. 
Subalgebra-subregion duality guarantees that for every domain of dependence in the bulk there is an associated boundary algebra~\cite{LeuLiu22}.  Generally, though, there is no clear way to ascertain which factor corresponds to a given domain of dependence, and in particular, it may not have a geometric interpretation in the boundary theory. $\sA_{\sW_R}$ can be identified with the large $N$ limit\footnote{Note that this limit is state-dependent.} of the boundary subalgebra $\sX_R$ associated with $R$. 
The boundary identifications of $\sA_O$ and $\sA_I$ are more subtle. We will denote the boundary subalgebras identified with them respectively as 
\be
\sX_C = \sA_I, \quad \sX_S = \sA_O 
\ee
and equation~\eqref{des1} can also be written as 
\be \label{comD}
\sX_Q = \sX_C \lor \sX_S  \ .
\ee

The Python's lunch proposal implies that $\sX_S$ is a subalgebra consisting of only simple operators (i.e. those with sub-exponential complexity) and $\sX_C$ is a subalgebra of exponentially complex operators. It is  striking that gravity ``predicts'' that 
simple and complex operators form type III$_1$ algebras and the full algebra $\sX_Q$ associated with $Q$ can be  generated by 
a union of them. The structure~\eqref{comD} is universal and generically applicable to any subsystem with a non-minimal QES. 

We conclude this subsection with some further remarks on the nature of  $\sX_S$ and $\sX_C$, clarifying the connection between this algebraic language and the code subspace formulation:

\ben 

\item We say an operator $A$ is simple if it is given by sums of finite (in the $N \to \infty$ limit) products of single-trace operators or a limit of such operators.\footnote{For example, operators of the form $e^{i  \int d^d x\, f(x) {\cal O} (x)}$ where $f$ is an $O(1)$ real function and ${\cal O} (x)$ is a single-trace Hermitian operator can be obtained as a limit and thus are also simple operators.} 
A single-trace operator is dual to a bulk elementary field. Thus simple operators are related to finite products of bulk fields or their limits. $\sX_S$ consists of such operators. 

\item When $Q$ is edgeless, from the original bulk geometry $M$, it is possible to construct a new geometry $M'$  using simple operations (i.e. acting by bulk fields and boundary evolutions by $O(1)$ times), where the simple wedge $\hat O$ coincides with the causal wedge of $Q$~\cite{EngPen21a}. This gives an explicit demonstration that $\sA_O$ can be reconstructed from $Q$ sub-exponentially.  
Since we can CPT conjugate around $\chi_{\rm outermost}$, there also exists a geometry $M_{\rm simple}$ with the same $\hat{O}$ but where $\hat{O}=\sW_{Q}$. Physically, $M_{\rm simple}$  describes a coarse-grained state $\psi_{\rm simple}$ of the original state of the system~\cite{EngWal17b, EngWal18}. 

\item $\sX_Q = \sA_{\sW_Q}$ is obtained by taking the large $N$ limit of the operator algebra $\sB^{(N)}_Q$ in $Q$~\cite{LeuLiu22}. 
The limit is subtle and state-dependent: for different semiclassical states, different sets of operators in $\sB^{(N)}_Q$ 
may have well-defined correlation functions in the limit, and thus different sets survive the limit. 
The subset of simple operators that are dual to bulk fields always survive, independent of the specific states. But in $\sX_Q$ there can be complex operators generated by modular flows that are highly state-dependent. $\sA_I$ includes those. 

We can also define $\sA_I = \sX_C$ intrinsically in the boundary theory as follows $\sX_C = \sX_Q \cap \sX_S'$, where $\sX_S$ is the algebra of simple operators completed in the weak limit. 
We expect that in the large-$N$ limit this definition of $\sX_C$ captures complex operators as defined via the code subspace formalism.

Note that whether an operator is complex also depends on the choice of $Q$. An operator that can only be generated by modular flows in $Q$ could be a simple operator in a larger subsystem.

\een

\subsection{What happens at the Page time: dynamical transfer of complexity} \label{sec:pady}

Let us now return to the  evaporating black hole first discussed in Sec.~\ref{sec:quan1}. Consider first a time $t_{c} < t < t_P$, when the empty set $\chi_0$ is the minimal QES and there is another locally but not globally minimal QES $\chi_1$ located slightly behind the horizon. 
Taking a Cauchy slice $\Sigma_B$ of $\sW_{B(t)}$ passing through $\chi_1$, we denote the region interior and exterior to $\chi_1$ respectively as $I$ and $O$, i.e. 
\be 
\Sigma_{B} = I \cup O, \quad t < t_P \ .
\ee
From the discussion of last subsection, $\hat I=D[I]$ is a Python's lunch and $\hat O=D[O]$ is a simple wedge, and we have  
\be \label{ex1}
\sA_{B} = \sA_O \lor \sA_I, \quad \sA_R = \sA_{R_{\rm Hawk}} , \quad t < t_P
\ee
and $\sA_O$ is a simplex factor and $\sA_I$ is a complex factor. While~\eqref{cx} can always be defined, $C$ becomes exponentially large in $1/G_N$ only after some later time scale $t_{c1} > t_c$, which we expect is a finite fraction of $t_P$.  In the current case, $\log \sC (t) \sim \frac{1 }{2}
\frac{A_{0}-A_{1}}{4 G_{N}}$ where $A_{0}$ is the area of the black hole at the beginning of evaporation and $A_{1}$ is the area of the black hole at time $t$. So for there to be an ${\cal O}(G_{N}^{-1})$ difference and thus exponential complexity in $G_{N}$, the difference in areas must be be ${\cal O}(1)$ at $t_{c_1}$. We will always consider $t > t_{c1}$ below. 
In the second equation we have also copied here~\eqref{yer}. As discussed in Sec.~\ref{sec:aev}, $R_{\rm Hawk}$ corresponds to the part of the radiation that can be expressed in terms of radiated  modes (e.g. gravitational radiation) outside the horizon -- the part typically associated to Hawking's calculation and its ignorance of unitarity. Thus~\cite{AkeEng22} $\sA_{R_{\rm Hawk}}$ consists of simple operators. At a time $t_P > t > {\rm max} (t_b, t_{c1})$, $\sA_R = \sA_{R_{\rm Hawk}}$ is hence a simple type III algebra. 

Now consider $t> t_P$, which marks the time at which $\chi_1$ becomes the minimal QES for both $B(t)$ and $R(t)$. We then have 
\be \label{ex2}
\sA_{B} = \sA_O, \quad \sA_R = \sA_{R_{\rm Hawk}} \lor \sA_{I}, \quad t > t_P
\ee
 where in the second equation we have also copied here~\eqref{yer0}. 
 Region $\hat I$ now belongs to $\sW_{R(t)}$ and is a Python's lunch for $R (t)$.
 So now $\sA_I$ is a complex type III$_1$ factor of $\sA_R$, and $\sA_{R_{\rm Hawk}}$ remains a simple type III$_{1}$ factor. 
 That $\sA_I$ lies in ${\cal A}_{R}$  is essentially the Hayden-Preskill process~\cite{HayPre07}; that its reconstruction is complex is simply the statement of Harlow-Hayden~\cite{HarHay13} that decoding the Hawking radiation is a high complexity task.

Now comparing~\eqref{ex1} and~\eqref{ex2}, we see that at the Page time $t_P$, there is a transition arising from a transfer of complexity factor $\sA_I$ from the black hole system $B$ to the reservoir system $R$. All our previously noted aspects of change in the physics of the system at $t_P$ can be understood from the transfer of complexity:\footnote{The conclusions of the first two items below are of course not new. Here we merely reexamine them from the current perspective.}

\ben 

\item 
Due to this exchange, the entanglement between $B$ and $R$ transitions from that between $R_{\rm Hawk}$ and $I$ -- which is proportional to the size of $R_{\rm Hawk}$ -- to that between $O$ and $I$, which is proportional to the size of the black hole. This leads to the turnover as the size of $R_{\rm Hawk}$ increases with time while that of the black hole horizon decreases with time. 

\item This is immediately related to the black hole becoming transparent:  the difficulty of reconstructing the interior of a black hole before the Page time is transferred into the difficulty of decoding the Hawking radiation after the Page time. 

\item Both before and after the Page time, $R_{\rm Hawk}$ and $I$ are connected by a quantum wormhole, and $O$ and $I$ are connected by a classical wormhole. Due to the exchange of $\sA_I$, the connection between $B$ and $R$ transitions from being between $R_{\rm Hawk}$ and $I$ to being between $I$ and $O$.

\een  

It is natural to ask why only the complexity factor is exchanged between $B$ and $R$. The reason is simple: the physics of simple factors can be directly probed by low-energy subexponential complexity  operations, which have a smooth limit in the $G_N \to 0$ limit. Such dynamics must be smooth with time in this limit.  By contrast, probing $I$ requires exponentially complex (in $1/G_N$) operations, which do not have to have a smooth $G_N \to 0$ limit with time. In fact, we will argue below that the common discontinuous `phase transition' between dominant QESs is a consequence of the phenomenon of complexity transfer, and is therefore rather ubiquitous. 
 
\subsubsection*{The Two-Sided Evaporating Black Hole's Canonical Purification}

 \begin{figure}[t]
\begin{center}
\includegraphics[width=\textwidth]{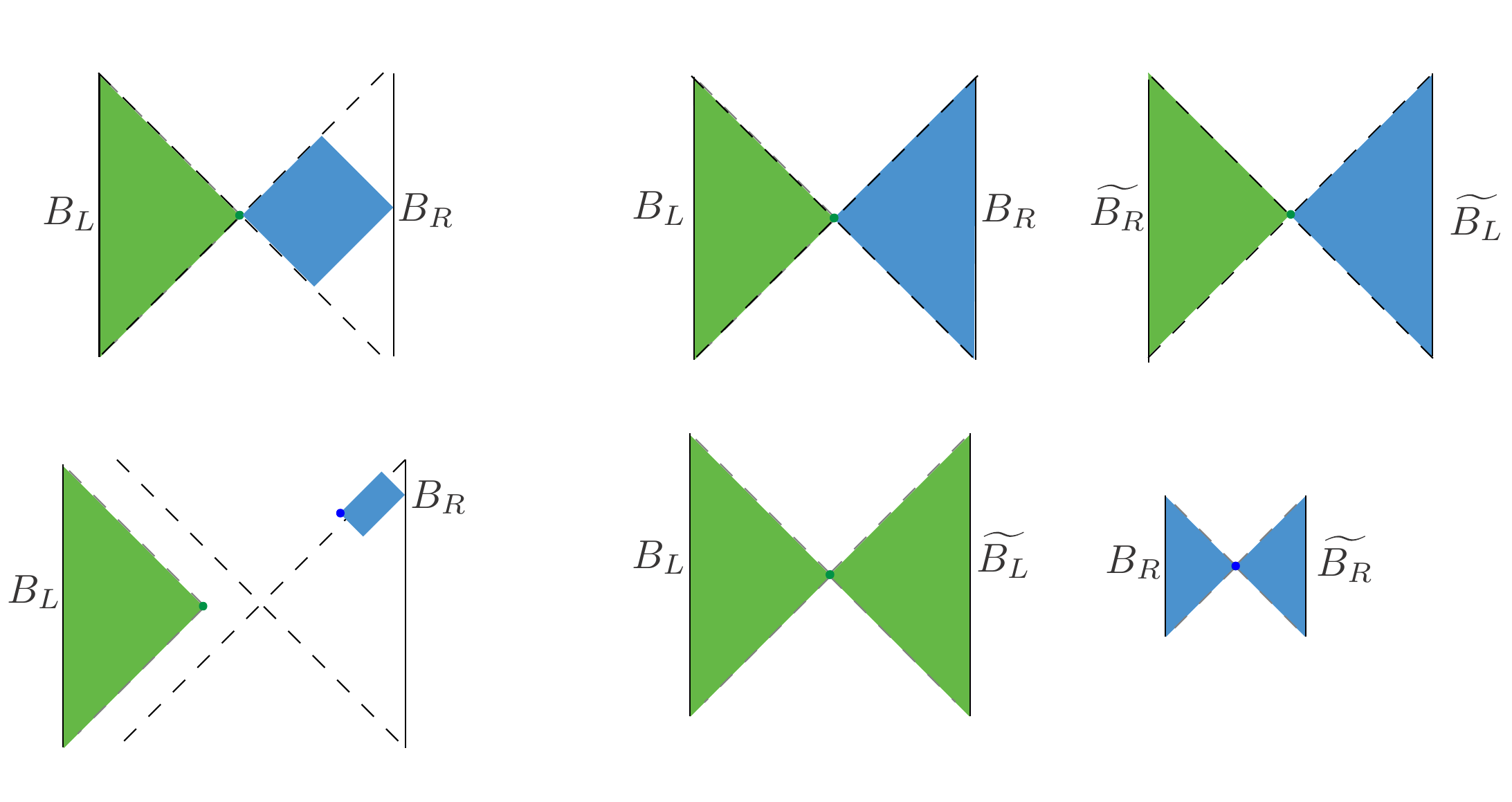} 
\end{center}
\caption{Canonical purification of the evaporating two-sided black hole before the Page time (top) and after the Page time (bottom). Here as usual we decouple from the path and evolve with the decoupled Hamiltonian on $B_{L}B_{R}$. In the former, $B_{L}$ and $B_{R}$ are connected; in the latter, $B_{R}$ and $\widetilde{B_{R}}$ are connected, and similarly for $B_{L}$ and $\widetilde{B_{L}}$.}

 \label{fig:twosidedCP}
\end{figure} 

As another example of the complexity factor transfer and its application in a multipartite setting, let us discuss a four-party version of the two-sided evaporating black hole first   originally presented in~\cite{EngFol22} (where the detailed construction of this spacetime may be found). We consider the same setup as in AEMM, this time without an end of the world brane to make it one-sided. That is, we have a two-sided black hole, where one side is at some finite boundary time coupled to a reservoir, and then trace out the reservoir. The result is a mixed state on $B_L$, the left boundary, and $B_R$, the right boundary.

Consider now building a four-party system via the canonical purification of $\psi_{B_{L}B_{R}}$ into $\ket{\psi_{B_{L}B_{R}\widetilde{B_{L}B_{R}}}}$.\footnote{To be clear, here we will be interested in the canonical purification in its own right, not as a tool towards understanding $\psi_{B_{L}B_{R}}$'s properties. }  In our discussion below, we will fix the time $t_L$ for $B_L$ to be $O(1)$, and consider how the connectivity changes with the time $t_R$; the time below will always refer to $t_R$.

We now present the connectivity puzzle in this context: before the Page time, the QES of $\psi_{B_{L}B_{R}}$ is the empty set, and $\sW_{B_{R}\cup B_{L} \cup \widetilde{B_{L}\cup B_{R}}}$ consists of two spacetimes that are classically disconnected, illustrated in Fig.~\ref{fig:twosidedCP}: $B_{R}$ and $B_L$ are classically connected to one another and similarly $\widetilde{B_R}$ and $\widetilde{B_L}$. After the Page time, there are still two classically disconnected spacetimes, but now the two spacetimes classically connect $B_R$ to $\widetilde{B_R}$ and $B_L$ to $\widetilde{B_L}$. However, we can pick $t_R$ at the two times so that the von Neumann entropies of any of the subsystems are identical. Thus even though classical connectivity swaps between the two pairs, the von Neumann entropy fails to provide an accurate diagnosis. 

If we consider $\ket{\psi_{B_{L}B_{R}\widetilde{B_{L}B_{R}}}}$ as a bipartite state on the systems $B_{R}\widetilde{B_{R}}$ and $B_{L}\widetilde{B_{L}}$, we can examine connectivity of these two systems from the algebraic perspective: the algebra of each system is type III$_1$ both before and after the transition, as long as we are at $t_R\sim {\cal O}(G_{N}^{-1})$. (The same is true if we consider $\ket{\psi_{B_{L}B_{R}\widetilde{B_{L}B_{R}}}}$ as a bipartite state on the systems $B_{R}B_{L}$ and $\widetilde{B_{L}}\widetilde{B_{R}}$ instead.)
Algebra type is similarly unable to distinguish the difference in connectivity. However, the spacetime type is distinguishing: at $t<t_P$, $B_R$ and $\widetilde{B_{R}}$ are quantum volatile while $B_L$ and $\widetilde{B_{L}}$ are classical. Because the canonical purification of $B_{L}$ is classical, the connectivity between $B_L$ and $B_R$ is also classical. Since $\widetilde{B_{R}}$ and $B_R$ are \textit{both} quantum volatile, they are only quantum connected. Finally, $B_L$ and $\widetilde{B_{L}}$ are disconnected. By contrast, shortly after $t_P$, $B_{L}$ and $B_{R}$ are disconnected (since $B_{R}$ remains classical, and classical spacetimes cannot be quantum connected); $B_{L}$ and $\widetilde{B_{L}}$ are classically connected; and $B_{R}$ and $\widetilde{B_{R}}$ are classically connected. 

What microscopically accounts for the myriad of different ways in which the four systems are connected or disconnected? As in the simpler case of a single-sided black hole, the answer is complexity transfer. Before the Page time, there is a high complexity factor in $B_{R}$. 
After the Page time, the complexity factor is transferred to the radiation and is gone from the subsystem $B_L\cup B_R$. 
As a  result the canonical purification after the Page time has only simple operators on each side.

\subsection{General formulation} 

We now show that a generic\footnote{The genericity alluded to here is to exclude cases where the difference in generalized entropies between the bulge and the appetizer is an ${\cal O}(1)$ number of bits.}  exchange of dominance of QESs can be understood physically as a consequence of the transfer of a complex factor, so the transition at the Page time discussed above  is representative of a more general phenomenon.

Consider a boundary subsystem $R_\alpha$ that depends on some continuous parameter $\alpha$ and has 
two QES candidates for the entanglement wedge: $\chi_0 (\alpha),\  \chi_1 (\alpha)$. These QESs compete for dominance, and for definiteness we will take $\chi_1$ to lie outside of  $\chi_0$.  
Suppose that at $\alpha = \alpha_0$, the two families exchange dominance. That is, $\chi_1$ is minimal for $\alpha > \alpha_0$, and $\chi_0$ is minimal for $\alpha<\alpha_{0}$. 
For $\alpha < \alpha_0$, there is a Python's lunch $I$ in $\sW_{R_\alpha}$, which is lost to $\overline{R_{\alpha}}$ for   $\alpha > \alpha_0$. As discussed earlier, by the strong Python's lunch proposal, a Python's lunch gives rise to a complex factor on the boundary, thus depending on the direction in which $\alpha_0$ is crossed, the operator algebra $\sA_R$ for $R_\alpha$ either gains or loses a type III$_1$ complex factor. 
Conversely, the inclusion or exclusion of a type III$_{1}$ factor of high complexity operators signals exchange of dominance of QESs. 
We formulate this more formally as follows.

\paragraph{Proposal:} Let $\{R_{\alpha}\}$ be a continuous family of boundary subsystems, which may include one or more complete connected $\mathscr{I}$'s. 
Whenever there is an exchange of dominance between QESs homologous to $R_{\alpha}$ as $\alpha$ is changed continuously, ${\cal A}_{R_{\alpha}}$  changes to either gain or lose a type III$_{1}$ factor of high complexity operators. Conversely, whenever ${\cal A}_{R_{\alpha}}$ gains or loses a type III$_{1}$ factor of high complexity operators, there is an exchange of dominance between QESs homologous to $R_{\alpha}$. 

\vspace{0.3cm}

The reason that only a factor of complex operators is exchanged is the same as that stated at the end of the previous subsection: only physics concerning simple operations has to be continuous in $\alpha$ in the $G_N \to 0 $ limit. On the gravity side this is reflected in continuous evolution of the outermost QES, even though the minimal QES may jump discontinuously. More explicitly, 
the outermost classical extremal surface always evolves continuously since the equation of extremal deviation is elliptic, so small perturbations of the boundary conditions result in small perturbations of the surface. This rigorous argument does not always survive for QESs, since the equation of quantum extremal deviation is integro-differential~\cite{EngFis19}; nevertheless we expect that typically, without large changes to the system, the outermost QES will likewise evolve smoothly.

\subsection{An example: two intervals in AdS\texorpdfstring{$_3$}{\_3}} 

Consider any two non-overlapping intervals $R_{1}, R_{2}$ on the boundary of a static time slice of global AdS$_{3}$ (with bulk quantum fields in their global vacuum\footnote{So that a QES in this geometry is in fact also a classical extremal surface.}). This is a classical spacetime (in particular, satisfying the classical condition). The entanglement wedge $\sW_{R_{1}\cup R_{2}}$ has two candidate QESs: the union of geodesics that begin and end on the same $R_{i}$, which we shall call $\chi_{\rm disc}$ , and the union of geodesics that begin on one $R_{i}$ and end on the other, which we shall call $\chi_{\rm conn}$. See the top panel of Fig.~\ref{fig:twointervals}.\footnote{Note that the complementary intervals have the same QES due to complementary recovery in the absence of quantum corrections: when $R_{1}\cup R_{2}$ is in the disconnected phase, the complement is in the connected phase and vice versa.}

\begin{figure}
    \centering
    \includegraphics[width=0.7\textwidth]{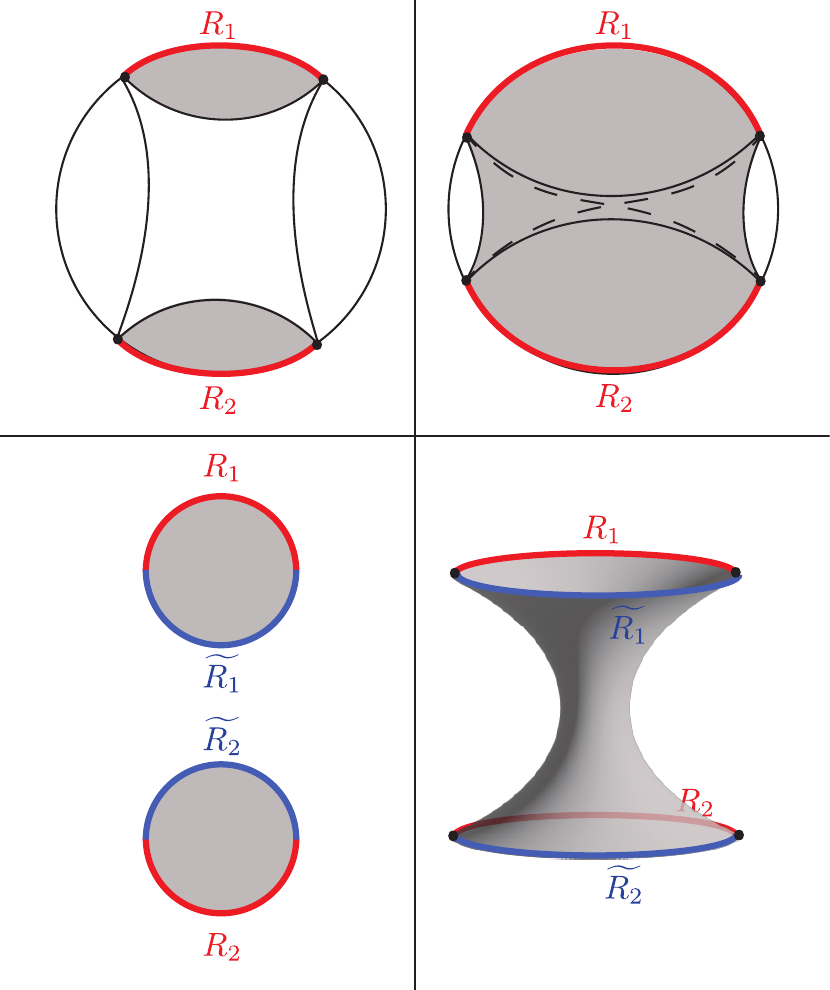}
    \caption{Cauchy slices of the two-interval transition in pure AdS$_{3}$. Top left: the disconnected phase of $R_{1}\cup R_{2}$. Top right: the connected phase of $R_{1}\cup R_{2}$ with the bulge surface shown as two dashed lines. Bottom left: the canonical purification of $R_{1}R_{2}$ in the disconnected phase is two disconnected spacetimes. Bottom right: the canonical purification of $R_{1}R_{2}$ in the connected phase is a single connected spacetime.}
    \label{fig:twointervals}
\end{figure}

When the intervals are small, $\chi_{\rm disc}$ dominates. The entanglement wedge does not connect $R_{1}$ to $R_{2}$. Only one QES lives in $\sW_{R_{1}\cup R_{2}}$ in this regime: reconstruction of operators in $\sW_{R_{1}\cup R_{2}}$ is simple, and the algebra of operators is the algebra of simple operators at large-$N$. As the size of the $R_{i}$ is gradually increased (while maintaining $R_{1}\cup R_{2}=\varnothing$), eventually the second QES $\chi_{\rm conn}$ becomes advantageous; $\sW_{R_{1}\cup R_{2}}$ becomes connected. There is a nonminimal QES -- $\chi_{\rm disc}$ -- in $\sW_{R_{1}\cup R_{2}}$: reconstruction of some operators is exponentially complex. The bulge surface is illustrated in the top right panel of Fig.~\ref{fig:twointervals}. 

Let us examine our connectivity proposal in this example. Since the state on $R_{1}\cup R_{2}$ is mixed, we must first canonically purify the state -- call it $\psi_{R_{1}R_{2}}$ -- into $\ket{\psi_{R_{1}R_{2}\widetilde{R_{1}R_{2}}}}$. When $\chi_{\rm disc}$ dominates, we CPT conjugate about $\chi_{\rm disc}$, which results in two disconnected spacetimes; when $\chi_{\rm conn}$ dominates, the requisite CPT conjugation results in a single connected spacetime. See the bottom panel of Fig.~\ref{fig:twointervals}.

We immediately find that ${\cal A}_{R_{1}\widetilde{R_{1}}}$ cannot be type III$_{1}$ when $\chi_{\rm disc}$ dominates and \textit{must} be type III$_{1}$ when $\chi_{\rm conn}$ dominates. There are many ways to see that it also cannot be type II in the former case (it is obviously not type II in the latter unless we make use of the crossed product construction): the simplest way is to simply note that the geometry is essentially pure AdS$_{3}$, which easily admits a pure state. 

Thus the algebra ${\cal A}_{R_{i}\widetilde{R_{i}}}$ is type I when $\sW_{R_{1}\cup R_{2}}$ is disconnected and type III$_{1}$ when it is connected, precisely matching our proposal.

As already noted above, we see the complexity factor exchange at play as well. For small intervals, reconstruction is simple: it is just HKLL~\cite{HamKab05, HamKab06, HamKab06b}. Once the wedge becomes connected and the relevant algebra type changes, there is a division into complex and simple factors, and no simple reconstruction exists for all operators~\cite{EngWal17b, EngPen21a}: our precise diagnostic of a QES exchange of dominance.

\section{Discussion} \label{sec:diss}

We have proposed a precise algebraic formulation of ER=EPR, including a rigorous definition of the ``quantum wormhole'' heuristically discussed in~\cite{MalSus13}. Our proposal associates spacetime connectivity with the structure of boundary operator algebras in the large-$N$ limit. 
This significantly illuminates the spacetime structure of an evaporating black hole, including the sense in which the interior of an old black hole is connected to the radiation through a quantum wormhole. 
We further argued that the exchange in the algebra of high complexity operators between a system and its complement is generically responsible for exchanges of dominance between QESs.

Here we first discuss some apparent counterexamples to our proposal and then end with some future perspectives.

\subsubsection*{Apparent counterexamples}

\begin{enumerate}
    \item \textbf{End-of-the-world brane:} We can use the setup of AEMM for a single-sided black hole by incorporating an (unflavored) end of the world brane (EOW) as in~\cite{KouMal17}. For this construction to work, it must be the case that even though the brane cuts off part of the spacetime, it does not incur a type III$_{1}$ algebra on a Cauchy slice as a result. While this has not been rigorously demonstrated, previous discussions based on a number of different arguments (see e.g.~\cite{Kol23}) have suggested that an EOW does not have the effect of introducing an internal spacetime boundary with a resulting type III$_{1}$ algebra but rather maintain the type I structure of the complete Cauchy slice of the spacetime without the brane. 
  
    \item \textbf{Traversable wormhole~\cite{GaoJaf16, MalQi18}}: these constructions involve \textit{causal} contact between two different asymptotic boundaries rather than spatial connectivity. 
   In this case the algebras on the two boundaries are not independent; in fact, each boundary algebra should already contain the full algebra, which should be type~I.
    \item \textbf{Pure JT gravity:} even though the two boundaries in pure JT gravity (canonically quantized) are connected,
   there is only a single copy  of a type I algebra, i.e. the two boundaries do not have independent algebras.
   Naively this appears to be a contradiction to our proposal of algebraic ER=EPR, since there is no type III$_{1}$ algebra involved at all. However, pure JT gravity has no local bulk degrees of freedom (though it has a single nonlocal degree of freedom)~\cite{HarJaf18}.\footnote{We thank D. Harlow for discussions on this point.} This canonically quantized bulk theory is not holographic, and we do not expect to have a standard notion of subalgebra/subregion duality as a result. Upon addition of matter, the algebra becomes type II at finite $G_{N}$~\cite{PenWit23}, which is outside of our regime of study, though it is generally consistent with our observations. 
\end{enumerate}

\subsubsection*{Future perspectives}

\ben

\item {\bf Distinguishing type III$_{1}$ algebras for quantum and classically connected cases.}
In the evaporating black hole example, $\sA_B$ is type III$_1$ both at $t_1$ and $t_2$ when the black hole and radiation are respectively quantum and classically connected. We can distinguish the two cases using the classical condition on the states $\psi_{B \widetilde B}$ obtained via canonical purification on the black hole subsystem (see Fig.~\ref{fig:cp} of Appendix~\ref{sec:grpu}). At $t_1$, the state does not satisfy the classical condition while at $t_2$ it does. It would be desirable to understand whether there is an intrinsic description in terms of properties of $\sA_B$ alone, i.e. whether there are some distinguishing features between $\sA_B$ at $t_1$ and $t_2$ that are responsible for the quantum and classical connectivity. 

This is also closely related to the following question:

\item  {\bf A diagnostic of a nontrivial area term in $S_{\rm gen}$.} 

Because the QES is computed using the generalized entropy, which is better defined than either of its individual components (see~\cite{BouFis15a} for a review), it has proven difficult to identify a quantitative measure that can discriminate between a trivial QES and a nontrivial one; though see~\cite{JafKol19, BelCol21} for progress in certain specific examples. Having a precise understanding of when the boundary entanglement is geometrized into a surface and when it is not is a clear step towards  a complete description of the way in which spacetime emerges from the dual theory.

In the evaporating black hole example, at $t_2$, the generalized entropy $S_{\rm gen}$ has an area term, but that at $t_1$ does not. 
With the classical condition imposed on a state,  boundary algebra being type III$_1$ predicts a nontrivial QES and the associated area term. When quantum volatile spacetimes are included, such a diagnostic is still lacking and clearly requires going beyond the algebra types to more refined algebraic properties.

\item  {\bf Mechanism for complexity transfer.}

At the Page time, a high complexity algebra is transferred from the black hole to the reservoir. We have argued this is a generic phenomenon whenever there is an exchange of dominance of QES. It would be ideal to have some microscopic understanding of the transfer mechanism. Here we mention an intriguing possible connection with the discussion of~\cite{LiuVar19,LiuVar20}, where it was argued that in chaotic lattice spin systems the Page turnover and the counterpart of exchange of dominance of QES are realized through a mechanism called void formation. 
More explicitly, under time evolution an operator can develop a void, where its nontrivial parts become separated by a
region of identity operators. For example, in the calculation of the Renyi entropy for the radiation subsystem in a toy model for black hole evaporation,  evolution of operators originally in the black hole subsystem to operators with support only in the radiation subsystem (i.e. with void formation in the black hole subsystem) become dominant at the Page time and lead to the Page 
turnover. It would be interesting to understand more precisely, perhaps in more elaborate models, whether 
void formations could be used to shed light on the complexity transfer.

\item  {\bf Information transfer from quantum wormholes?}

It is generally believed that information transfer from the black hole to the reservoir happens after the Page time, which as we discussed in Sec.~\ref{sec:pady} can also be understood naturally using the algebraic language as a consequence of classical connectivity between the black hole and the reservoir entanglement wedges. Since before the Page time the black hole and the reservoir are connected by a quantum wormhole, it is natural to ask whether such quantum connection can facilitate information ``transfer'' from the black hole to the reservoir. In~\cite{VarKud21a,VarKud21b}, it was found from examining the behavior of logarithmic negativity that  there are significant entanglement correlations within the radiation long before the Page time. Such entanglement correlations are rather subtle, cannot be distilled using LOCC (local operations and classical communications), and the physical origin of such correlations is not clear. Relatedly, the proposed resolution of~\cite{AkeEng22} of the black hole information paradox involves teleportation of the information facilitated by a non-isometric encoding map.
Could the subtle entanglement correlations in the radiation be related to the quantum wormhole connecting it to the black hole? We hope to examine this in the future.

\item  {\bf Inclusion of stringy corrections?}

While our boundary discussion hitherto has only been concerned with the large $N$ limit, the bulk gravity discussion requires in addition that the stringy corrections are suppressed.\footnote{For theories such as the $M$-theory on AdS$_4 \times S_7$ or AdS$_7 \times S_4$, where $G_N$ is the only bulk parameter, such questions do not arise.} 
 It is desirable to understand how stringy corrections affect our proposal, which may shed light on the nature of ``stringy'' geometries. In particular, it is possible that the algebraic definitions that we have used here can be used to ``define'' what it means to be ``classically connected'' (in the sense that there is no quantum gravitational fluctuations) and ``quantum connected'' even in the regime where stringy corrections are large and the usual geometric descriptions do not apply.

\item  {\bf Emergent Type III$_{1}$ Non-Geometric Subalgebras} It is worth emphasizing a novel aspect of our construction illustrated by the subalgebra of complex operators: while in the bulk the algebra of simple operators and the algebra of operators behind the lunch at large-$N$ corresponds to two different bulk regions, in the boundary these subalgebras are non-geometric and emergent only at large-$N$. This unusual feature appears to be a special of holographic theories with a geometric bulk dual.

\item \textbf{Potential Applications Beyond AdS} While our results are concretely formulated in the context of AdS/CFT, it is not unreasonable to expect that some aspects of our approach and perspective can be applied to say, asymptotically flat or dS spacetimes. See for example~\cite{FraPar23a, FraPar23b, Gom23} for recent discussions of ER=EPR in such contexts.

\een

\section*{Acknowledgments}
It is a pleasure to thank C. Gomez, D. Harlow, A. Karch, J. Kudler-Flam, S. Leutheusser, J. Maldacena, D. Marolf, N. Paquette, J. Sorce, A. Speranza, and E. Witten for helpful discussions. The work of NE is in part by the U.S. Department of Energy under Early Career Award DE-SC0021886, by the John Templeton Foundation via the Black Hole Initiative, by the Sloan Foundation, by the Heising-Simons Foundation, and by funds from the MIT physics department. The work of HL is supported by the Office of High Energy Physics of U.S. Department of Energy under grant Contract Number  DE-SC0012567 and DE-SC0020360 (MIT contract \# 578218).

\appendix

 \section{Canonical purification and the Evaporating Black Hole} \label{sec:grpu}

Here we consider the canonical purification of the evaporating black hole example discussed in the Introduction (recall Fig.~\ref{fig:bulkwedges} and Fig.~\ref{fig:count}). Readers should consult Sec.~\ref{sec:cano} first for a review of canonical purification.
We consider the canonical purifications of the black hole subsystem $B$ at times $t_1$ and $t_2$ respectively. As described above,  the canonical purification of $B(t)$ is obtained by  CPT conjugating $\sW_{B(t)} $ around its dominant QES. This introduces a second copy of $B$, denoted $\widetilde{B}$, and results in a pure state 
on $B \widetilde{B}$.
At $t_1$ the dominant QES is $\chi_0=\varnothing$, so Cauchy slices of $\sW_{B(t_1)}$ are inextendible, and thus  the canonical purification of $B (t_1)$ is disconnected, as indicated in the top panel of Fig.~\ref{fig:cp}.   At $t_2$, the canonical purification of $B(t_2)$ is shown in the bottom panel of Fig.~\ref{fig:cp}, where we have a connected geometry. 
Through the canonical purification, now all subsystems have a geometric description, and at $t_1$, we have two disconnected geometries which are entangled at $\sO(1/G_N)$. 
 This provides an explicit demonstration of the puzzle discussed in Sec.~\ref{sec:intro}  in which all subsystems are geometric.

 \begin{figure}[t]
\begin{center}
\includegraphics[width=0.9\textwidth]{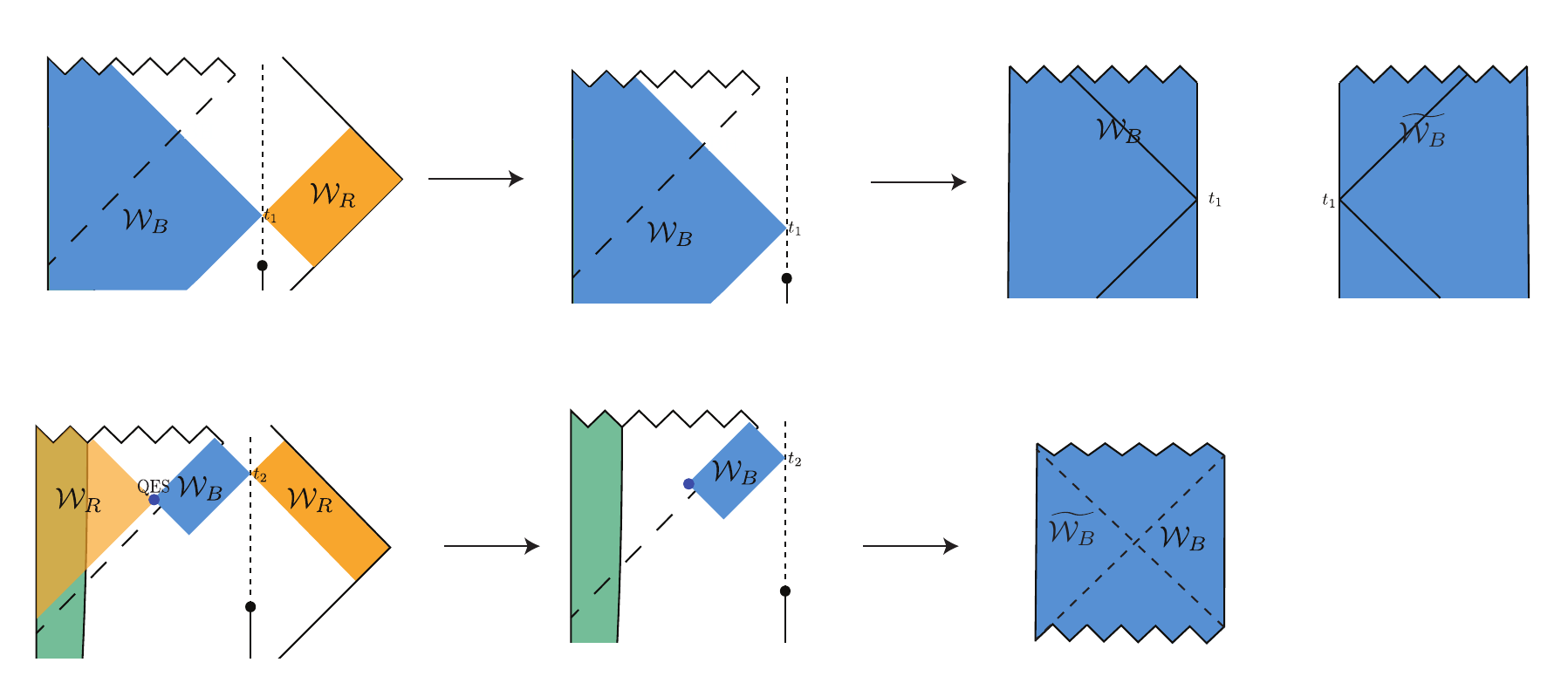} 
\end{center}
\caption{Canonical purification of $B(t)$: we decouple $B(t)$  from the bath, canonically purify, and evolve with the decoupled Hamiltonian of $B(t)\cup \widetilde{B(t)}$. Top: at $t_1$ we obtain disconnected spacetimes for $B$ and $\widetilde B$. Bottom: at $t_2$, we obtain a connected spacetime.
}
 \label{fig:cp}
\end{figure}

\bibliographystyle{jhep}
\bibliography{all}

\end{document}